\documentclass{DCMreport}
\usepackage{graphicx}
\usepackage{hyperref}	
\usepackage{cite}
\usepackage{breakurl}

\usepackage{amsmath,amssymb,amsfonts,amsthm}
\graphicspath{{Diagrams/}{Plot/}}

\usepackage[caption=false]{subfig}

\usepackage{fixltx2e}

\usepackage{bm}
\usepackage{booktabs}

\usepackage{cases}

\usepackage{verbatim}
\usepackage[usenames]{color}
\usepackage{datetime}

\usepackage{algorithm}
\usepackage{algpseudocode}

\newtheorem{thm}{Theorem}

\DeclareMathAlphabet{\pzc}{OT1}{pzc}{m}{it}

\newcommand{\mc}[1]{\ensuremath{\mathcal{#1}}}

\usepackage{geometry}

\title{Designing and Embedding Reliable Virtual Infrastructures}

\author{Wai-Leong Yeow, C\'{e}dric Westphal and Ula\c{s} C. Kozat}

\email{wlyeow@ieee.org, \{cwestphal,kozat\}@docomolabs-usa.com}
\address{DoCoMo USA Labs\\3240 Hillview Ave, Palo Alto, CA 94304, USA}

\reportnumber{DCL-TR-2010-15}
\date{Wednesday 31$^\text{st}$ March, 2010}

\makeatletter
\hypersetup{
	pdftitle={\@title},    
	pdfauthor={\@author},     
	pdfkeywords={Reliability, Network Virtualization, Algorithms, Design}, 
	hypertexnames=false,
	plainpages=false
}
\makeatother

\begin{document}

\maketitle

\begin{abstract}
In a virtualized infrastructure where physical resources are shared, a single physical server failure will terminate several virtual servers and crippling the virtual infrastructures which contained those virtual servers. In the worst case, more failures may cascade from overloading the remaining servers. To guarantee some level of reliability, each virtual infrastructure, at instantiation, should be augmented with backup virtual nodes and links that have sufficient capacities. This ensures that, when physical failures occur, sufficient computing resources are available and the virtual network topology is preserved. However, in doing so, the utilization of the physical infrastructure may be greatly reduced. This can be circumvented if backup resources are pooled and shared across multiple virtual infrastructures, and intelligently embedded in the physical infrastructure. These techniques can reduce the physical footprint of virtual backups while guaranteeing reliability.
\end{abstract}

\tableofcontents

\section{Introduction} 
\label{sec:intro}

With infrastructure rapidly becoming virtualized, shared and dynamically changing, it is essential to provide strong reliability to the physical infrastructure, since a single physical server or link failure affects several shared virtualized entities. Providing reliability is often linked with over-provisioning both computational and network capacities, and employing load balancing for additional robustness. Such high availability systems are good for applications where large discontinuity may be tolerable, e.g. restart of network flows while re-routing over link or node failures, or partial job restarts at node failures. A higher level of fault tolerance is required at applications where some failures have a substantial impact on the current {\em state} of the system. For instance, virtual networks with servers which perform admission control, scheduling, load balancing, bandwidth broking, AAA or other NOC operations that maintain snapshots of the network state, cannot tolerate total failures. In master-slave/ worker architectures, e.g. MapReduce, 
failures at the master nodes waste resources at the slaves/workers.

Through synchronization~\cite{Bressoud96,Brendan08} and migration techniques \cite{Clark05,Wang08} on virtual machines and routers, we postulate that {\em fault tolerance can be introduced at the virtualization layer}. This has several benefits. Different levels of reliability can be customized and provisioned over the same physical infrastructure. There is no need for specialized, fault tolerant servers. Instead, redundant (backup) virtual servers can be created dynamically, and resources are pulled together, increasing the primary capacity. Both will lead to a better overall utilization of the physical infrastructure.

In this paper, we propose an Opportunistic Redundancy Pooling (ORP) mechanism to leverage the properties of the virtualized infrastructure and achieve a $n:k$ redundancy architecture, where $k$ redundant resources can be backups for {\em any} of the $n$ primary resources, and {\em share} the backups across multiple virtual infrastructures (VInfs).

For a quick motivating example, consider two VInfs with $n_1$ and $n_2$ computing nodes. They would require $k_1$ and $k_2$ redundancy to be guaranteed reliability of $r_1$ and $r_2$, respectively. Sharing the backups will achieve a redundancy of $k' = \max(k_1,k_2)$ with the same level of reliability, reducing the resources that are provisioned for fault tolerance by at most 50\%.

In addition, there is joint node and link redundancy such that a redundant node can take over a failed node with guaranteed connectivity and bandwidth. ORP ensures VInfs do not connect to more redundant nodes than necessary in order to keep the number of redundant links low.


The other contribution of this paper is a method to statically allocate physical resources (compute capacity and bandwidth) to the primary and redundant VInfs simultaneously, taking into account the output of the ORP mechanism. It attempts to reduce resources allocated for redundancy by utilizing existing redundant nodes, and overlapping bandwidths of the redundant virtual links as much as possible.

Our paper focuses on the problem of resource allocation for virtual infrastructure embedding with reliability guarantee. Practical issues such as system health monitoring, protocol design, recovery procedures, and timing issues are out of the scope of this paper.

The organization of this paper is as follows. In the next section, we briefly describe the background, notations and define reliability in~\autoref{sec:ps}. Then, we describe a virtual architecture that can provide fault tolerance and estimate the benefits of sharing redundancies in~\autoref{sec:redgain}. We see how the link topology is preserved under failures in~\autoref{sec:preserve}, and how resources can be efficiently allocated in the physical infrastructure in~\autoref{sec:resAlloc}.
Finally, we evaluate and validate the ideas through simulation in~\autoref{sec:eval}, present related work in~\autoref{sec:related} , and~\autoref{sec:con} concludes this paper.


\section{Problem Statement} 
\label{sec:ps}

We consider a resource allocation problem in a virtualized infrastructure, such as a data center, where the virtualized resources can be leased with reliability guarantees. The physical infrastructure is modeled as an undirected graph $\mc{G} = (\mc{N}, \mc{E})$, where $\mc{N}$ is the set of physical nodes and $\mc{E}$ is the set of physical links. Each node $\mu \in \mc{N}$ has an available computational capacity of $\Gamma_\mu$. Each undirected link $(\mu,\nu) \in \mc{E}, \mu, \nu \in \mc{N}$ has an available bandwidth capacity of $\Lambda_{\mu\nu}$.

Each resources lease request is modeled as an undirected graph $G = (N, E)$. $N$ is a set of compute nodes and $E$ is a set of edges. We call this a virtual infrastructure (VInf). $\gamma_u$ is the computation capacity requirement for each node $u \in N$, and bandwidth requirements between nodes are $\lambda_{uv}$, $(u,v) \in E$ and $u, v \in N$. 

Reliability is guaranteed on the set of critical nodes $C \subseteq N$ of a VInf $G$ through redundant virtual nodes in the physical infrastructure $\mc{G}$. A backup (redundant) node $b$ must be able to assume full execution of a failed critical node $c$. Hence, the backup node must have sufficient resources in terms of computation $\gamma_b \ge \gamma_c$ and bandwidth to neighbors of $c$: $\lambda_{bu} \ge \lambda_{cu}, \forall u \in N, (c,u) \in E$.

The problem is, thus, to allocate as little resources as possible for a VInf $G$ on a physical infrastructure $\mc{G}$, including redundancy such that a reliability guarantee of at least $r$ is achieved. We explain the definition of reliability in greater detail in the next section.


\section{Reliability} 
\label{sec:reliability}

We define reliability as the probability that critical nodes of a VInf remain in operation, over all possible node failures. This is not to be confused with availability, which is defined as a ratio of uptime to the sum of uptime and downtime~\cite{reliabilitybook}. As an example, the reliability of a physical node under a renewal process is
\begin{equation}
	1 - \frac{1}{\text{MTBF}} ,
\end{equation}
whereas the availability is
\begin{equation}
	\frac{\text{MTBF}}{\text{MTBF + MTTR}},
\end{equation}
where MTBF is the mean time between failures as specified by the manufacturer, and MTTR is the mean time to recover from a failure. Hence, by guaranteeing a reliability of $r$, we are ensuring that there are sufficient redundant physical resources available in times of failure, with probability $r$. For a VInf with $n$ critical nodes and $k$ backup nodes, we want to ensure that
\begin{equation}
	Pr( k \; \text{out of}\; n + k \; \text{virtual nodes fail}) < 1 - r .
\end{equation}
This covers cases where some critical and backup nodes fail simultaneously. Guaranteeing availability, on the other hand, is ensuring that the system MTTR is low enough with respect to the system MTBF.

\subsection{Failover configuration} 
\label{sec:failover}

While provisioning redundant resources is a fundamental approach to guaranteeing reliability, there are two main classes of configurations in dealing with the redundant resources:
\begin{description}
	\item[Active/Active] All available nodes, including the $k$ redundant ones, are online, with an external load balancer distributing load across them. When physical failures happen, the load balancer can redistribute the load to the remaining nodes. This inherently assumes that (i) all nodes are homogeneous, and (ii) load can be distributed easily among all nodes.
	\item[Active/Passive] Redundant nodes are kept idle. When $x$ virtual nodes fail, $x$ redundant nodes are activated to take over. In the simplest setup, $n$ sets of $1:m$ backups are required. Nodes can be heterogeneous but the number of redundant nodes is restricted to $k = nm$. However, if each redundant node is capable of assuming the operation of more than one critical node, it is then possible to have $k < n$.
\end{description}

In the Active/Active configuration, the number of redundant nodes $k$ is the minimum. In comparison, $k$ is actually the same as that of the Active/Passive configuration for the same level of reliability if each redundant node has sufficient resource to assume the operations of any $n$ critical node. This is because the load in Active/Active configuration can always be shifted to all $n$ critical nodes, leaving all $k$ redundant nodes idle.

With Active/Passive configuration where critical nodes are heterogeneous, all redundant nodes must have all $n$ critical virtual machine (VM) images on disk (but not running in memory) in order to minimize $k$. This means that when a backup takes over, the critical VM's state is essentially rebooted, i.e., redundant nodes are ``cold spares''. On the other hand, the Active/Active configuration will not be able to support heterogeneous VMs.

It is possible to have the redundant nodes as ``hot spares'' in the Active/Passive configuration. That will require active synchronization techniques such as Remus~\cite{Brendan08}, ample memory in the physical node and bandwidth within the network for all $n$ states and their updates, respectively. A technique called Difference Engine~\cite{DiffEngine09} is able to reduce the physical footprint of the $n$ memory states, by keeping only one copy of the similar pages between the $n$ memory states and selectively compressing the remaining differences.

Redundant nodes can be further pooled and shared across several VInfs in the Active/Passive configuration since heterogeneous VMs can be supported, and further reducing redundant resources. This gives the Active/Passive configuration another advantage over the other. More details will be explained in \autoref{sec:redgain}. For the rest of this paper, we assume the Active/Passive configuration is used due to the reduction in redundancy.


\subsection{How many backups?} 
\label{sec:numbackup}

The number of redundant nodes depend on the physical mapping, and the failure models of both the physical nodes and the virtual infrastructure. For example, assume all $n$ critical nodes and all $k$ backups are placed on physical nodes A and B, respectively. Let $F_x$ be some set where $x$ out of the $n$ critical virtual nodes has failed. Then, the reliability of this example system is
\begin{multline}
	r = Pr(\text{A works and B fails})  Pr( F_0 ) + \mbox{}\\
	Pr(\text{A fails and B works}) Pr( n \le k ) + \mbox{}\\
	Pr(\text{A and B works }) 
	\sum_{x=1}^{\min(n,k)} \sum_{F_x \subseteq \{1, \ldots, n\} } Pr( F_x ) .
\end{multline}
There are several problems with this approach:
\begin{enumerate}
	\item computing the reliability is complex due to tight correlation between physical and virtual nodes.
	\item the reliability is severely limited by the reliability of the two physical nodes, and
	\item the reliability can never be increased beyond $k = n$,
\end{enumerate}
As such, we impose two physical mapping constraints: (i) each virtual node is only mapped to one physical node, and (ii) the mapped physical nodes are placed apart to avoid correlated failures among the physical machines, e.g. on different racks with different power supplies. This way, the failure rate of a virtual node $u$ is directly derived from the physical node $\mu$ it is mapped on. Guaranteeing reliability can then be focused on the failure model of the virtual infrastructure. In general, the reliability of the overall virtual infrastructure (including $k$ redundant nodes) can then be computed as
\begin{equation}
	\label{eqn:genericReliability}
	\begin{split}
		r(k) &= \sum_{y=0}^k Pr( \text{$y$ out of $n+k$ nodes fail} ) \\
		&= \sum_{y=0}^{k} \sum_{x=0}^{\min(n,y)} Pr( \text{a set of $y-x$ backup nodes fail} ) f(x) ,
	\end{split}
\end{equation}
for some failure probability distribution $f(x)$ in which $x$ is the number of critical nodes that failed. \eqref{eqn:genericReliability} can be simplified to
\begin{equation}
	\label{eqn:simpReliability}
	r(k) = \sum_{y=0}^{k} \sum_{x=0}^{\min(n,y)} {k \choose y-x} p^{y-x} (1-p)^{k-(y-x)} f(x) .
\end{equation}
The binomial term is due to independent failures of the redundant nodes, and the assumption that physical nodes hosting them are homogeneous with a failure rate $p$.

\subsubsection{Cascading Failures} 
\label{sec:cascade}

It is also possible to compute the reliability for cascading failure models of these critical nodes. There is a wealth of studies on various cascading failure models in literature. In this paper, we list three models and briefly describe how the innermost sum $\sum_{F_x \subseteq \{1, \dotsc, n\}} Pr(F_x)$ of \eqref{eqn:genericReliability} can be computed. For a more detailed discussion, please refer to the appendix.
\begin{description}
	\item[Load-based~\cite{DobsonI05}] This model assumes a node will fail if its load exceed a predefined value. Once some nodes fail, the load on other nodes are incremented with a value proportional to the number of failures. The failure cascades if more nodes fail from the overloads. The main result from this model is the distribution of the total number of failures $x$, which can directly replace the term $f(x)$.
	\item[Tree-based~\cite{IyerSM09}] This model uses a continuous-time Markov Chain (CTMC) to analyze cascading failures. A node failure will stochastically cause nodes from other categories to fail. There is a renewal repair process for each node, as well as redundant nodes for each node category. A procedure is given to compute the generator matrix $Q$ for the CTMC, which is used as a basis for analyzing various reliability metrics of the system. For our purpose, we can obtain $Q$ of the system {\em without any backups} by setting the renewal rate to follow the behavior of MTTR, and the number of redundant nodes per category to $0$. Subsequently, $Pr(F_x)$ is a direct one-to-one mapping to the steady-state probabilities, which can be obtained by solving the null-space of $Q$.
	\item[Degree-based~\cite{WattsDJ02}] Each node has a predefined failure threshold between 0 and 1. The VInf is initially perturbed with some random node failures. The failure will cascade to a neighboring node if the neighbors of that node that have failed is beyond its failure threshold. For a large $n$, a global cascading failure will occur with some probability $f(n)$ if the average degree of the VInf is lower than some value. Since $f(x)$ is unknown for $x < n$, \eqref{eqn:simpReliability} uses a worst-case distribution for $f(x)$ where $f(n-1) = 1-f(n)$. We refer the reader to the appendix for more details.
\end{description}

Once $f(x)$ is obtained, a numerical method for searching $k$ can then be used to ensure guarantee a certain level of reliability $r$. With a binary search algorithm, the complexity is in the order of $O(n^2 \log n)$. A proof is given in the appendix.

\subsubsection{Independent Failures} 
\label{sec:indepFail}

For ease of exposition, we focus in this document on independent node failures.
If the failure rates of $n$ critical nodes are independent and uniform, then the reliability of the whole system is
\begin{equation}
	\label{eq:rIndep}
	\begin{split}
		r(k)	&= \sum_{x=0}^k { n+k \choose x } p^x (1-p)^{n+k-x} \\
			&= I_{1-p} (n, k+1) ,
	\end{split}
\end{equation}
where $I_{1-p} (\cdot,\cdot)$ is the regularized incomplete beta function with parameter $1-p$~\cite{handbook}. The minimum number of redundant nodes $k$ is then the integer ceiling of the inverse of this function.



\section{Redundancy pooling: quantization gains} 
\label{sec:redgain}

In Active/Passive configuration, redundant nodes that are provisioned for one VInf can be shared with another VInf, since they are idle. This is not possible in the Active/Active configuration because all virtual nodes will have to be running. The ability to share redundant nodes allows the Active/Passive configuration to reduce the amount of redundant resources within the physical infrastructure.

\subsection{Arbitrary Pooling} 
\label{sec:arbpool}

To simplify discussion, we first assume the case where failure rates of all critical and physical nodes are independent and uniform, i.e, $p_{u\mu} = p$ for all virtual nodes $u$ and its physical host $\mu$. \autoref{fig:scaleNK} shows the number of backup nodes required as the number of critical nodes increase for a reliability guarantee of $99.999\%$ over various failure probabilities $p = 0.01, \ldots, 0.05$. The range of failure values were chosen due to a recent Intel study on physical server failures in data centers of different locations and with different types of cooling \cite{intel}.

\begin{figure}
	\centering
	\includegraphics[width=.7\linewidth]{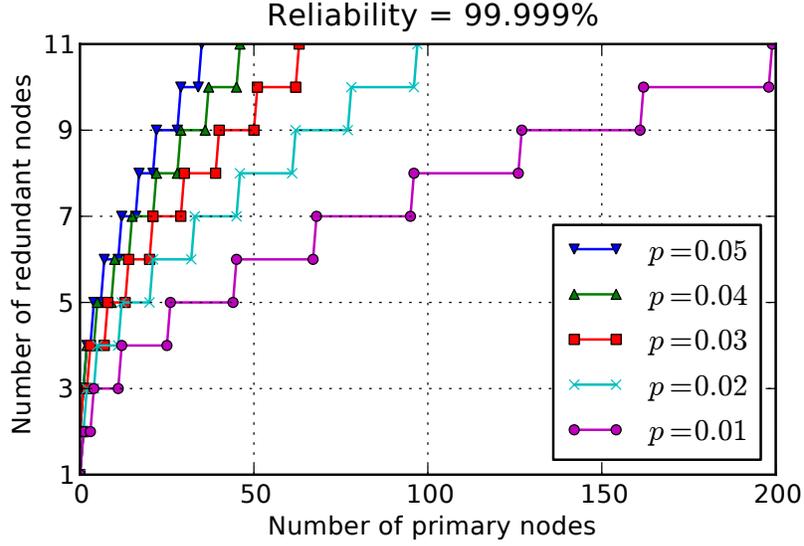}
	\caption{Minimum number of redundant nodes needed as backups for $n$ critical nodes, given a reliability guarantee of 99.999\%.}
	\label{fig:scaleNK}
\end{figure}

By observation from these curves, one intuitive way to reduce the number of backup nodes is to exploits the curves' sub-linear property: if two or more VInfs pool their backup nodes, the total number of backup nodes required is reduced. For example, for the case of $p = 0.01$, two VInfs of 100 critical nodes each will require a total of 16 backup nodes ($k = 8$ each). This number can be reduced to 11 when both VInfs (gives $n = 200$) pool their backup nodes together, saving redundant resources by 31.25\%. Unfortunately, this approach has two potential limitations: (i) the pooling advantage disappears for VInfs with large $n$ and (ii) redundant bandwidth maybe increased while reducing redundant backup nodes. In the former case, we have the following result:
\begin{thm}
	\label{thm:linearK}
	Given $n$ is related to $k$, $p$ and $r$ as in \eqref{eq:rIndep}. For large $k$, $\frac{n}{k}$ is a constant and is independent of $r$, i.e.,
\begin{equation}
	\lim_{k \to \infty} \frac{n}{k} = \frac{1}{p} - 1 .
\end{equation}
\end{thm}
\begin{proof}
	The function $I_{1-p}(n,k+1) = r$ is the CDF of a Binomial variable $\chi$, which characterizes the number of node failures from a pool of $n+k$ nodes, each with a failure probability $p$, i.e.,
	\begin{equation*}
		I_{1-p}(n,k+1) = P\{\chi \le k\} = r .
	\end{equation*}
	Since $n > 0$ and is non-decreasing as $k \to \infty$, by Strong Law of Large Numbers
	\begin{equation*}
		\chi = p(n+k) + \sqrt{p(1-p)(n+k)} \mathcal{N}_{0,1} ,
	\end{equation*}
	Then, for some constant $\xi_r$,
	\begin{align*}
		k &= p(n+k) + \xi_r \sqrt{p(1-p)(n+k)} \\
		\frac{k}{n+k} &= p + \xi_r \sqrt{\frac{p(1-p)}{n+k}} \\
		k \to \infty, \qquad \frac{k}{n+k} &= p \\
		\frac{n}{k} &= \frac{1}{p} - 1 .
	\end{align*}
\end{proof}
Since $k$ tends to be linear with $n$ for VInfs with a large number of critical nodes, the benefit of sharing backup nodes diminishes in these cases as the number of backup nodes will not be reduced further.

The second limitation is that more redundant links maybe required when pooling backup nodes. We use the same example where two VInfs with 100 critical nodes each can share the pool of 11 backup nodes to illustrate this. If the two VInfs were not sharing backup nodes, each VInf would have 8 backup nodes instead of 11. Since each backup node must be able to resume full execution of the critical nodes, each backup node will need additional redundant links to neighbors of all critical nodes\footnote{The number of additional links is $O(nk + k^2)$ for $n$ critical nodes and $k$ backup nodes. The $nk$ term is due to links between neighbors of all $n$ critical nodes to the $k$ backups, and the $k^2$ term is due to links between backup nodes. We describe this in detail in \autoref{sec:preserve}}, as described in \autoref{sec:ps}. Thus, each VInf will need an additional three sets of redundant links with redundancy pooling. Savings in computational resource in this case may not justify the increase in bandwidth reserved for redundancy. Furthermore, this increases overhead in a hot standby configuration as each VInf has more backup nodes to synchronize to.

\autoref{fig:mapRshare} illustrates this tradeoff when two VInfs (VInf-1 and VInf-2) arbitrarily pool the backup nodes with $p = 0.02$. Denote by $n_1$ and $k_1$ the number of critical nodes of VInf-1 and the minimum number of backup nodes required for reliability of 99.999\%, respectively. The same notation applies to VInf-2: $n_2$ and $k_2$. Further, denote by $k'$ the minimum number of backup nodes required for a VInf with $n_1 + n_2$ critical nodes. With the super-linear property, it is guaranteed that $k' < k_1 + k_2$ for finite values. But depending on the relation of $k'$ with $k_1$ and $k_2$ individually, three regions of tradeoff can be expected:
\begin{description}
	\item[More links in both VInfs: $\bm{k' > k_1}$ and $\bm{k' > k_2}$.] Both VInfs need more backup nodes each with redundancy pooling, as illustrated in the previous example. This translates to additional costs of more redundant bandwidth used, and synchronization overhead if backup nodes are hot standbys, while reducing redundant computational resource.
	\item[No tradeoff: $\bm{k' = k_1 = k_2}$.] Both VInfs do not need more backup nodes each with redundancy pooling. This is the ideal case as the total number of backup nodes is halved through pooling, and there is no increase in redundant links nor synchronization overhead. However, this region is small and only exists for small $n$, indicating that there is not much opportunity to pool backup nodes in this way.
	\item[More links in one VInf: $\bm{k' = k_1}$ or $\bm{k' = k_2}$, not both.] This is an intermediate region between the prior two cases. One VInf will need more backup nodes in order to pool backup nodes with the other VInf, whereas the other VInf is not affected.
\end{description}

\begin{figure}
	\centering
	\includegraphics[width=.85\linewidth]{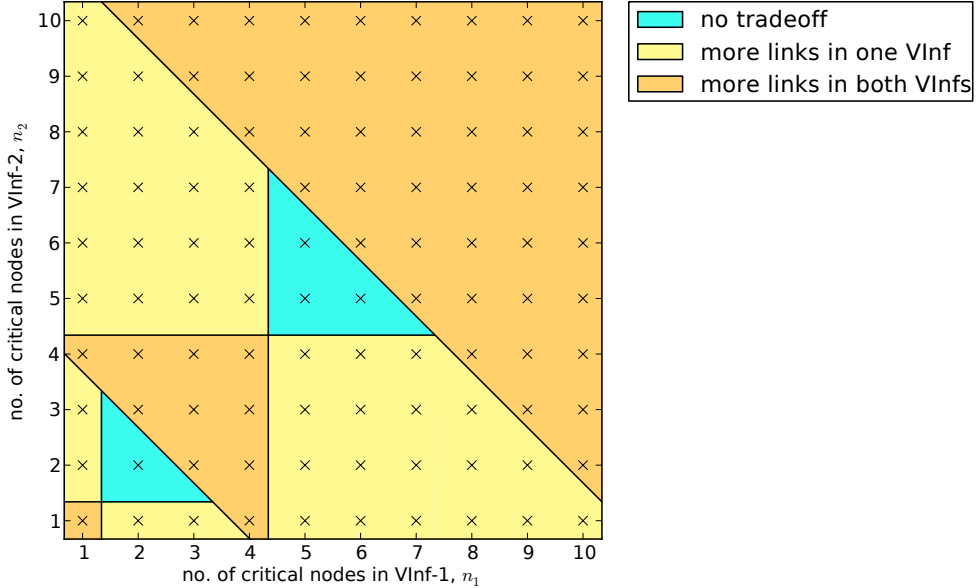}
	\caption{The regions of tradeoff between computational and bandwidth resource when pooling backup nodes of two arbitrary VInfs. The values are obtained from \autoref{fig:scaleNK} for the case of 99.999\% reliability and $p = 0.02$.}
	\label{fig:mapRshare}
\end{figure}


\subsection{Pooling to fill in the discrete gap} 
\label{sec:gapPool}

For this reason, we introduce Opportunistic Redundancy Pooling (ORP). This is a method to pool backup nodes such that there is no additional overhead on bandwidth (and synchronization, in the case of hot standbys). Another advantage with this method is that VInfs with different reliability guarantees can be pooled together. It makes use of the discrete steps of the curves as shown in \autoref{fig:scaleNK}. For example, in the case where $p  = 0.03$, a VInf with 29 critical nodes needs 7 backup nodes and the reliability evaluates to 99.999065093\%. Another VInf with one more critical node needs 8 backup nodes in order for the reliability to be maintained above the guarantee. In particular, the reliability for the latter case is 99.9998544522\%, which is much higher than the guarantee. It is also the reason why the number of backups is the same for $n \in [30,39]$. In the case of 30 critical nodes, the excess 0.0008544522\% reliability can be ``sacrificed'' to ``squeeze in'' other VInfs that require no more than 8 backup nodes\footnote{This means either the VInf needs lower reliability guarantee, has smaller number of critical nodes, has a skewed $f(x)$ that gives smaller $k$, or all of the above.}. Conversely, it can also be viewed that the residual excess from a few VInfs are pooled to reduce the number of backups one VInf needs.

\begin{figure}
	\centering
	\includegraphics[width=.6\linewidth]{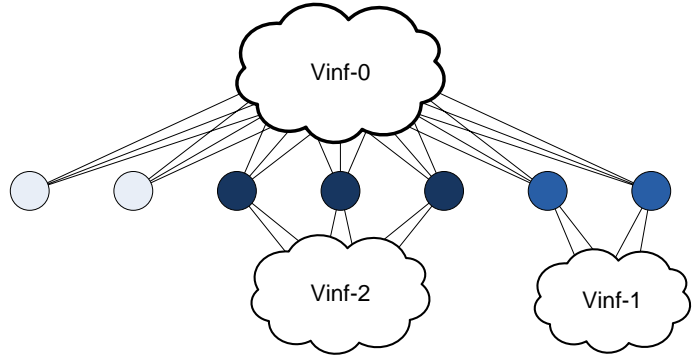}
	\caption{A way to pool backup nodes. This can be viewed in two ways: (i) a virtual infrastructure VInf-0 ``lends'' some of its backup nodes to other VInfs, and (ii) VInfs 1 and 2 collectively ``lend'' their backup nodes for VInf-0. The total number of backup nodes will be reduced by at most 50\%.}
	\label{fig:label}
\end{figure}

We use \autoref{fig:mapRshare} to explain ORP. Suppose there are $m+1$ VInfs which require $k_0, k_1, \ldots, k_m$ backup nodes for reliability guarantee of $r_0, r_1, \ldots, r_m$, respectively, and $k_0 \ge \sum_{i=1}^m k_i$. VInf-0 pools its backup nodes with VInf-$i$, $i > 0$, and each backup node backs up VInf-0 and at most one other VInf. Having this restriction allows us to keep the reliability guarantees of any VInf-$i$ satisfiable (i.e., the reliabilities of VInf-$i$ are no less than $r_i$ before and after pooling) and shifts the reliability computation to VInf-0. The decision criterion to pool the VInfs is then determining whether the new reliability of VInf-0 after pooling $r'_0$ is still greater than the required guarantee $r_0$. As described later, imposing this criterion supports an incremental evaluation of $r'_0$ when VInfs are newly admitted into the redundancy pool, or leave the redundancy pool. This also requires the assumption that the recovery protocol when activating backup nodes prioritize VInf-$i$ over VInf-0, and VInf-$i$ uses no more than $k_i$ backup nodes after recovery.

Define $z^\text{VInf}(k, y)$ as the probability that a total of $y$ nodes fail in a VInf with $k$ backup nodes, i.e.,
\begin{equation}
	z^\text{VInf}(k,y) = \sum_{x=0}^{\min(n, y)} {k \choose y-x} p^{y-x} (1-p)^{k-(y-x)} f^\text{VInf}(x) ,
\end{equation}
where $n$ is the number of nodes of that VInf. The reliability of VInf-0 after pooling is then
\begin{equation}
	\label{eqn:r0Reliability}
	\begin{split}
		r'_0 = 1 - \sum_{x=0}^{k'} & Pr(\text{$x$ of $k'$ backups are down or used by VInf-1, \ldots, VInf-$m$}) \times \\
		& Pr(\text{more than $k_0-x$ nodes fail from VInf-0 with $k_0-k'$ backups}) .
	\end{split}
\end{equation}
The first term is the probability mass function (pmf) of the sum of $m$ independent VInfs with $k_i$ backup nodes each. The pmf of each independent event $q^\text{VInf-$i$}(x)$ is
\begin{equation}
	q^\text{VInf-$i$}(x) = \begin{cases}
		z^\text{VInf-$i$}(k_i, x) & , 0 \le x < k_i \\
		1 - \displaystyle\sum_{y=0}^{k_i - 1} z^\text{VInf-$i$}(k_i, y) & , x = k_i \\
		0 & , \text{otherwise.}
	\end{cases}
\end{equation}
Convolving all $m$ pmfs give the first term of \eqref{eqn:r0Reliability} to be
\begin{equation}
	Q(x) = \mathcal{F}^{-1} \biggl( \prod_{i=1}^m \mathcal{F} \bigl( q^\text{VInf-$i$}(x) \bigr) \biggr) ,
\end{equation}
where $\mathcal{F}(\cdot)$ and $\mathcal{F}^{-1}(\cdot)$ is the Discrete Fourier Transform (DFT) and its inverse, respectively, of minimum length $k'$. It is, however, more convenient to keep the length to be at least $k_0$ so that more VInfs can be pooled in future without having to recompute $m$ DFTs again\footnote{For performance reasons, the length could be kept at $2^{\lceil \log_2 k_0 \rceil}$ to take advantage of Fast Fourier Transform algorithms.}. Then, \eqref{eqn:r0Reliability} simplifies to
\begin{equation}
	\label{eqn:r0SimpR}
	r'_0 = 1-\sum_{x=0}^{k'} Q(x) \biggl( 1 - \sum_{y=0}^{k_0-x} z^\text{VInf-0}(k_0-k', y) \biggr) .
\end{equation}
The time complexity to decide whether VInfs $1, \ldots, m$ can be pooled with VInf-0 is bounded by the $m$ DFTs, which evaluates to $O(mk \log k)$.

\begin{figure}
	\centering
	\includegraphics[width=.8\linewidth]{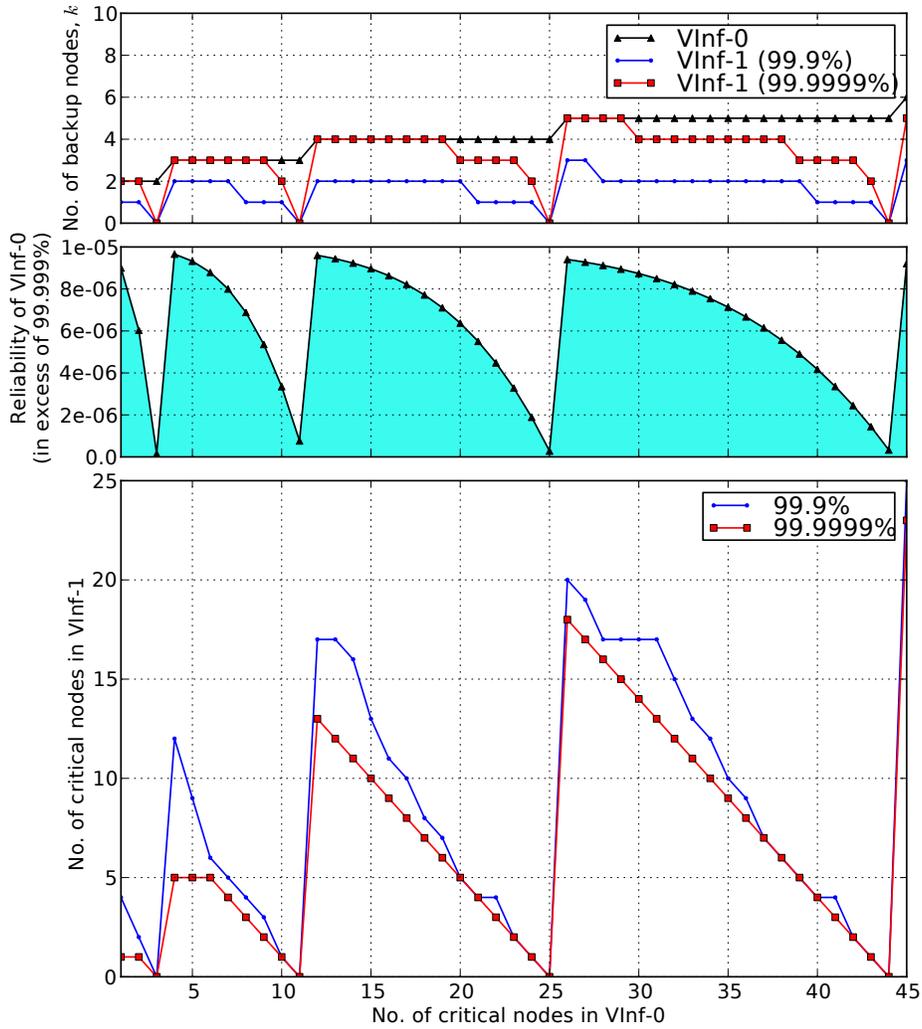}
	\caption{VInf-0 shares its $k_0$ backup nodes with VInf-1, where $p=0.01$. The lower plot shows the maximum number of critical nodes VInf-1 can have while reusing VInf-0's backup nodes for two cases of reliability guarantee, while maintaining VInf-0's reliability above a requirement of 99.999\%. The upper plot shows the corresponding number of VInf-0's backup nodes used by VInf-1, and the middle plot shows the excess reliability from VInf-0 that has been tapped upon for VInf-1. Hence, the shape of the curves are similar.}
	\label{fig:gapShare}
\end{figure}

Refer to \autoref{fig:gapShare} for a graphical explanation. In this figure, we show the case where VInf-0 shares some of its backup nodes with one other VInf under the current pooling scheme. We show two cases of VInf-1, one with reliability requirement of 99.9\% and another with 99.9999\%. VInf-0's reliability requirement is kept at 99.999\%, in between those two cases. For simplicity, failure rates of all critical nodes are set to be independent and uniform at $p=0.01$.

In the top plot, the number of backup nodes $k_0$ for VInf-0 increases in a step-wise fashion similar to \autoref{fig:scaleNK} as the number of critical nodes increase. The step-wise increase in $k_0$ creates opportunities for VInf-1 to reuse some of the backup nodes VInf-0 have since there is much excess in VInf-0's relibability prior to pooling, as shown in the shaded area in the middle plot. The lower plot shows the maximum number of critical nodes VInf-1 can have in order to pool backup nodes with VInf-0, and the respective number of backup nodes reused are shown in the top plot. Since VInf-1 is essentially utilizing VInf-0's excess reliability, the peaks and valleys of the curve in the lower plot follows that of the middle plot. It can be observed, too, that the size of VInf-1 is significant as compared to that of VInf-0, and the number of backup nodes conserved is up to 50\%.

The advantage of this pooling scheme can be summarized as follows.
\begin{description}
	\item[No tradeoff.] The shared VInfs use only the minimum number of backups as though there is no pooling. Hence, there are no additional links than required.
	\item[Does not diminish for large $\bm{n}$.] The pooling scheme makes use of the excess reliability arise from discrete steps in $k$. Hence, there will always be gaps that can be filled with VInfs that need smaller $k$.
	\item[Pooling over different $\bm{r}$.] This scheme allows for VInfs of arbitrary reliability requirements to be pool together.
	\item[Flexibility in adding VInfs.] A new VInf-$m+1$ can always be added into VInf-0's pool of backup nodes so long as VInf-0's new reliability computed from \eqref{eqn:r0SimpR} is still above the required guarantee. As mentioned previously, all previous $m$ DFTs can be stored to speed up this admission control procedure by a factor of $O(m)$ to $O(k \log k)$.
	\item[Flexibility in removing VInfs.] VInf-$i$ can always be removed from the pool since VInfs other than VInf-0 are unaffected, and VInf-0 will have its effective reliability increased. Conversely, if VInf-0 is to be removed, the other VInfs simply reclaim the respective backup nodes as their own, which can be pooled with new incoming VInfs.
\end{description}

There can be other ways of extending this pooling scheme. For example, VInf-$i$ shares its backup nodes to another lower layer of VInfs, and recursively add new layers. Another example is to have a new VInf-$i$ share across two VInf-0s. We do not study these cases due to three major reasons: (i) there is a compromise of flexibility in dynamically adding and removing VInfs, (ii) the gains may be marginal as compared to the initial sharing, and (iii) the time complexity to re-evaluate of the reliabilities of all pooled VInfs may be high.

Pooling VInfs this way is opportunistic, since we do not predict the statistics of future incoming VInfs. In general, however, VInf-0 should be the one with the largest number of backup nodes as this allows for more degrees of freedom in choosing other VInf-$i$ to pool backup nodes with.



\section{Preserving Virtual Infrastructure} 
\label{sec:preserve}

The virtual infrastructure $G=(V,E)$ has to be preserved when backup nodes resume execution of failed critical nodes. This translates to ensuring that every backup node has guaranteed bandwidth to all neighbors of all critical nodes.

\subsection{Minimum Redundant Links} 
\label{sec:optGraph}

It is possible to minimize the total number of links while providing redundancy for a VInf. Harary and Hayes \cite{HararyF96} studied the problem of constructing a new graph $G' = (N \cup B, E')$ with minimum links (i.e., $|E'|$) such that upon the removal of any $k = |B|$ nodes (i.e., $k$ node failures), the resultant graph always contain the original VInf $G$.

However, this poses a limitation since the result only guarantees graph isomorphism and not equality. In other words, there may be a need to physically swap remaining VMs while recovering from some failure in order return to the original infrastructure $G$. Recovery may then be delayed or require more resources are available for such swapping operations.

We illustrate this using the example in \autoref{fig:ex}. The new graph $G'$ in \autoref{fig:ex:opt} is obtained using Theorem~2 in \cite{HararyF96}\footnote{Theorem~2 works only on unweighted graphs. In the subsequent step, we obtained the minimum link weights through exhaustive iteration.}. If nodes $\mathsf{b_1}$ and $\mathsf{c_3}$ fail, the only way to recover is to have $\mathsf{c_3}$ to be in the position of the current $\mathsf{c_2}$ and backup node $\mathsf{b_2}$ assuming the role of $\mathsf{c_2}$. The problem here is that node $\mathsf{c_2}$ is not a backup for node $\mathsf{c_3}$ in the first place! Hence, the recovery procedure is lengthened to two steps: (i) recover node $\mathsf{c_3}$ at backup node $\mathsf{b_2}$, and then (ii) swap nodes $\mathsf{c_2}$ and $\mathsf{c_3}$. This problem will always arise no matter where the backup nodes are placed in $G'$.

Furthermore, deriving optimal graphs $G'$ with minimal links any general graph $G$ has exponential complexity. To the best of the authors' knowledge, optimal solutions are found only on regular graphs such as lines, square-grids, circles, and trees \cite{HararyF96,AjtaiM92,DuttS97}.

\begin{figure}
	\centering
	\subfloat[VInf with three critical nodes.\label{fig:ex:vinf}]{%
		\makebox[.22\linewidth][c]{\includegraphics{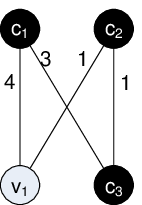}}}\hfil
	\subfloat[Least number of links for VInf with two backups.\label{fig:ex:opt}]{%
		\makebox[.32\linewidth]{\includegraphics{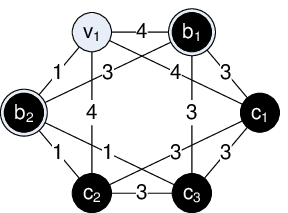}}}\hfil
	\subfloat[Addition of redundant (bold) links $L$ to VInf for two backups.\label{fig:ex:exlinks}]{%
		\makebox[.38\linewidth]{\includegraphics{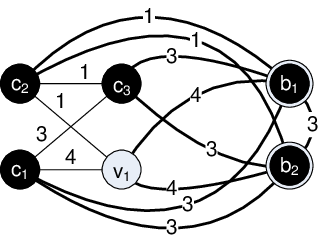}}}
	\caption{~\protect\subref{fig:ex:opt} and~\protect\subref{fig:ex:exlinks} show two different ways of preserving VInf~\protect\subref{fig:ex:vinf} to account for two redundant nodes. While the former utilizes 1 less link and 1 less bandwidth unit, it requires some nodes to be rearranged in the recovery phase.}
	\label{fig:ex}
\end{figure}


\subsection{Redundant Links without Swapping} 
\label{sec:reduce}

To overcome the aforementioned limitations, we choose to use the following set of redundant links, at the expense of incurring more redundant resources. Formally, the set of redundant links $L$ that are added to $G$ is
\begin{align}
	L &= L^1 \cup L^2 \\
	L^1 &= \{ (b,u) \,|\, \exists (c,u) \in E, \forall b \in B, u \in N, c \in C \} \\
	L^2 &= \{ (a,b) \,|\, \forall a, b \in B \} ,
\end{align}
where $B$ and $C$ are the sets of backup and critical nodes, respectively. $L$ is a union of two sets of links. The first set $L^1$ connects all backup nodes to all neighbors of all critical nodes, and the second set $L^2$ interconnects all backup nodes since two critical nodes may be neighbors of each other and may fail simultaneously. The latter set can be omitted if there are no links between any critical nodes, i.e.,
\begin{equation}
	L^2 = \emptyset \;\Leftrightarrow\; (c,d) \notin E, \forall c, d \in C .
\end{equation}

\autoref{fig:ex:vinf} and \autoref{fig:ex:exlinks} illustrates an example in the expansion of the edges of a VInf for the added redundancy. The former figure shows the original VInf of four nodes, in which three of them are critical. In providing two additional backup nodes for reliability, more links are added to the primary VInf so that each backup node has the full bandwidth resource as any critical node, as seen in the latter figure.

Adding $L$ to $G$ is much more straightforward and does not suffer from the aforementioned swapping / rearrangement problem. More importantly, when pooling backups across VInfs, redundant links $L$ can be added without affecting other existing VInfs which use the same backup nodes.

The number of redundant links in $L$ may seem large ($O(nk + k^2)$, where the first and second terms are the number of links in $L^1$ and $L^2$, respectively). But, the amount of physical bandwidth reserved can be reduced while embedding them. This is because not all links will be in use at the same time. A simple example in \autoref{fig:overlap} can illustrate this. The small VInf in \autoref{fig:overlap:vinf} consists of two nodes and a link between them with 1 bandwidth unit. One of the nodes is critical and is backed up by two redundant nodes. Suppose that due to limited available compute capacities, the physical deployment of the virtual nodes is that of \autoref{fig:overlap:map}. If the redundant links are embedded verbatim into the physical infrastructure, the link $(\mu_2,\mu_4)$ would require 2 units. However, it is only necessary to reserve 1 unit on this link, since at most 1 backup node will be in use at any time.

\begin{figure}
	\centering
	\subfloat[\label{fig:overlap:vinf}]{\includegraphics{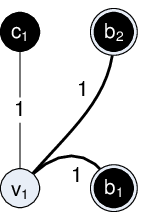}}\hfil
	\subfloat[\label{fig:overlap:map}]{\includegraphics{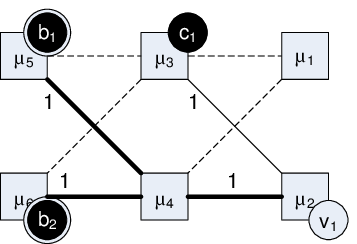}}
	\caption{On the left is a VInf with redundant links to two backup nodes. Suppose that the virtual nodes (circles) are deployed at the physical nodes (squares) on the right. Unutilized links are dotted lines and redundant links are in bold. Only 1 unit of bandwidth needs to be reserved on link $(\mu_2, \mu_4)$.}
	\label{fig:overlap}
\end{figure}

The solution is more complex in a scenario with a slightly larger VInf. \autoref{fig:exembed} is an example of embedding the VInf from \autoref{fig:ex:exlinks} with three different placements of the backup nodes.  Similarly, at most four out of the nine redundant links will be utilized at any time, i.e., two critical node failures, so the physical bandwidth reservation should only need to cater for all cases where the backup nodes recover two critical nodes. As can be seen in the figure, placement of the backup nodes and determining the minimal physical bandwidth on the set of redundant links is a complex problem. The amount of physical bandwidth required depends on the ways in which a redundant link can be ``overlapped'' with other redundant links, which in turn, depend on the physical location of the backup nodes. Due to this highly coupled relation between backup nodes and redundant links, evaluating placement of virtual nodes is much more complex than that of embedding a VInf without any redundancy.

\begin{figure}
	\centering
	\subfloat[\label{fig:ex:m1}]{\includegraphics[viewport=0 -6.78 72.24 114.96] {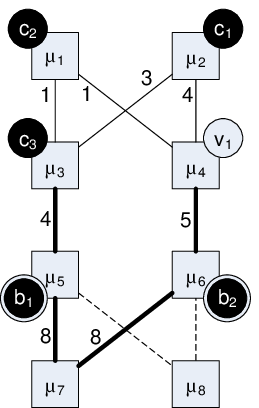}}\hfil
	\subfloat[\label{fig:ex:m2}]{\includegraphics{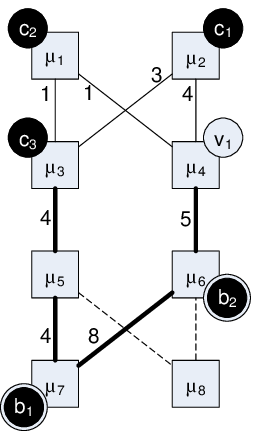}}\hfil
	\subfloat[\label{fig:ex:m3}]{\includegraphics{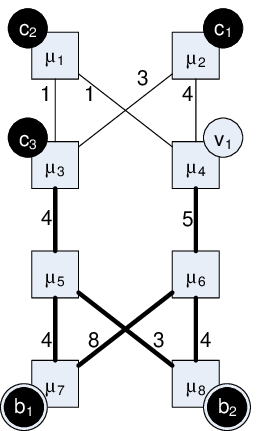}}%
	\caption{Some ways of embedding a VInf of \protect\autoref{fig:ex:vinf} with two backup nodes in a physical infrastructure. Circular and square nodes represent virtual and physical nodes, respectively, and dashed lines represent unutilized physical links. The mass of redundant links which connect the backups to the VInf (see \protect\autoref{fig:ex:exlinks}), can be minimally embedded since not all bold links are in use at any instance.}
	\label{fig:exembed}
\end{figure}



\section{Resource Allocation: a mixed integer programming problem}
\label{sec:resAlloc}

To address the tight coupling between the virtual nodes and links of a VInf, we use a joint node and link allocation approach for the resource allocation problem. In \cite{VineYard09}, a multi-commodity flow (MCF) problem is formulated to jointly allocate nodes and links of a VInf to physical infrastructure. We adapt the MCF problem with additional constraints to solve for the minimum bandwidth used on redundant links.

The MCF problem is a network flow problem where the objective is to assign flows between sources and destinations in a network. The virtual links of a VInf can be seen as flows between virtual nodes. To determine the actual locations of the virtual nodes, the physical network is appended with virtual nodes and ``mapping'' links connecting every virtual node to their possible physical locations, i.e., links $(u, \mu)$ for all virtual nodes $u$ and physical nodes $\mu$ are appended to the set of physical links $\mc{E}$. The first physical node in which a flow passes through will be the location of the virtual node of that flow (link). Additional constraints are added to the MCF to ensure that a virtual node has only one physical location.

\subsection{Mapping Constraints} 
\label{sec:mapConst}

Denote by $\rho_{u\mu}$ a binary variable that represents the mapping between a physical node and a virtual node, i.e., $\rho_{u\mu} = 1$ if a virtual node $u$ is mapped onto physical node $\mu$, 0 otherwise. The mapping constraints can then be expressed by the following equations:
\begin{gather}
	\label{eqn:oneNodeMap}
	\sum_{\mu \in \mc{N}} \rho_{u\mu} = 1 , \quad \forall u \in N \cup B , \\
	\label{eqn:noNodeOverlap}
	\sum_{u \in N \cup B} \rho_{u\mu} \le 1 , \quad \forall \mu \in \mc{N} .
\end{gather}
The two equations ensure that there is exactly one hosted virtual node among all physical nodes, and that a physical node can host at most one virtual node, respectively. More conditions can be added. For example, if the VInf is to reuse backup nodes from an existing pool of backups in physical nodes, then the mapping variables $\rho_{u\mu} = 1$ for every existing virtual backup and physical node pair $(u,\mu)$. The mapping restrictions can be extended further to three other scenarios.

\begin{description}
	\item[Location exclusion.] Some physical nodes may not be used with a VInf. For example, in the case where backup nodes are pooled, locations of the new virtual nodes should not be the same as any of the existing VInfs with the same backup nodes. This is necessary for the reliability in \eqref{eqn:r0SimpR} to be valid. Then, for the set of physical nodes $\mc{N}'$ that should be excluded,
	\begin{equation}
		\rho_{u\mu} = 0, \quad \forall u \in N,  \forall \mu \in \mc{N}' .
	\end{equation}
	\item[Location preference.] Some physical locations may be preferred for some virtual nodes. For example, load balancers, firewalls, ingress and egress routers, or could be as simple as physical proximity. Then, for a set of preferred physical locations $\Phi(u)$ for a virtual node $u$,
	\begin{equation}
		\rho_{u\mu} = 0, \quad \forall \mu \notin \Phi(u) .
	\end{equation}

	\item[Location separation.] In order to avoid correlated failures, critical nodes can be placed, for example, in different racks. For a set of physical nodes on the same rack $\mc{N}'$, the following constraint can be appended:
	\begin{equation}
		\sum_{\mu \in \mc{N}'} \sum_{u \in N} \rho_{u\mu} \le 1 .
	\end{equation}
\end{description}


\subsection{Resource Constraints} 
\label{sec:resConst}

Compute capacity constraints on the physical nodes can be easily captured through the mapping variables, i.e.,
\begin{equation}
	\rho_{u\mu} \gamma_u \le \Gamma_\mu , \quad \forall u \in N \cup B, \forall \mu \in \mc{N'} .
\end{equation}
The set of backup nodes $B$ may be omitted from the above if backup nodes are reused from a redundancy pool, and the compute capacity reserved is already more than the maximum of that of the new critical nodes, i.e, $\max_{c \in C} \gamma_c$. Otherwise, the above constraints may be included with the RHS equals to the deficit compute capacity.

Bandwidth constraints and link mappings are derived from the MCF problem. As mentioned earlier, a virtual link $(u,v)$ between two virtual nodes $u$ and $v$ can be seen as a flow between source and destination under the MCF problem. Due to the inclusion of redundant links, we define four types of flows:
\begin{itemize}
	\item Virtual links $E$: flows between two virtual nodes $u, v \in N$. The amount of bandwidth used on a link $(i,j)$ is denoted by $\ell^{uv}_E [ij]$.
	\item $L^1$ flows, i.e., flows between a backup node $a \in B$ and a neighbor $v \in N$ of some critical node. The amount of bandwidth on $L^1$ flows depends on which critical node $c$ the backup node $a$ recovers, and how much bandwidth can be ``overlapped'' across different failure scenarios. As such, we denote by $\ell_{L^1}^{acv} [ij]$ the amount of bandwidth used on a link $(i,j)$ when such a recovery occurs. This allows us to model the overlaps between redundant flows.
	\item Aggregate flows on a link between redundant nodes $B$ and the neighbor $v \in N$ of some critical node. This reflects the actual amount of bandwidth reserved after overlaps on link $(i,j)$. We denote this by $\ell_o^v[ij]$.
	\item $L^2$: flows between two backup nodes $a, b \in B$. The amount of bandwidth used on a link $(i,j)$ is denoted by $\ell_{L^2}^{ab}[ij]$. Unlike $L^1$ flows, we do not model any possible overlapping of these redundant links with $L^1$. This is to ensure the $L^2$ flows can be easily reused when sharing with other VInfs.
\end{itemize}

The flows $E$, $L^1$ and $L^2$ follow the conservation of flow equations. At the source, the respective flow constraints are:
\begin{gather}
	\sum_{\mu \in \mc{N}} \ell^{uv}_E [u\mu] - \ell^{uv}_E [\mu u] = \lambda_{uv} , \quad \forall (u,v) \in E , \\
	\sum_{\mu \in \mc{N}} \ell^{acv}_{L^1} [a\mu] - \ell^{acv}_{L^1} [\mu a] = \lambda_{cv} , \quad \forall (a,v) \in L^1, \forall c \in C , \\
	\sum_{\mu \in \mc{N}} \ell^{ab}_{L^2} [a\mu] - \ell^{ab}_{L^2} [\mu a] = \max_{\substack{c,d \in C, \\ (c,d) \in E}} \lambda_{cd} , \quad \forall (a,b) \in L^2 .
\end{gather}
The above equations state that the total flow out of the source nodes $u, a$ must be equal to the flow demand. For links in $E$ and $L^1$, the flow demand is the bandwidth requirement of that virtual link. For links in $L^2$, it is the maximum bandwidth requirement between any two critical nodes. In the case where backup nodes are pooled, we can omit reserving bandwidth for $L^2$ flows (and hence the $L^2$ flow conservation constraints) unless the existing reservation is insufficient. In that case, only the excess need to be reserved.

The flow conservation constraints at the destination is similar to that of the source:
\begin{gather}
	\sum_{\mu \in \mc{N}} \ell^{uv}_E [v \mu] - \ell^{uv}_E [\mu v] = -\lambda_{uv} , \quad \forall (u,v) \in E , \\
	\sum_{\mu \in \mc{N}} \ell^{acv}_{L^1} [v \mu] - \ell^{acv}_{L^1} [\mu v] = -\lambda_{cv} , \quad \forall (a,v) \in L^1, \forall c \in C , \\
	\sum_{\mu \in \mc{N}} \ell^{ab}_{L^2} [b \mu] - \ell^{ab}_{L^2} [\mu b] = -\max_{\substack{c,d \in C, \\ (c,d) \in E}} \lambda_{cd} , \quad \forall (a,b) \in L^2 ,
\end{gather}
except that the flow demand is negative as the flow direction is into that node.

At the physical intermediate nodes $\mu$, the flow conservation constraints for flows $E$, $L^1$ and $L^2$, respectively, are
\begin{gather}
	\sum_{i \in N \cup \mc{N}} \ell^{uv}_E [i \mu] - \ell^{uv}_E[\mu i] = 0, \quad \forall \mu \in \mc{N}, \forall (u,v) \in E , \\
	\sum_{i \in N \cup \mc{N}} \ell^{acv}_{L^1} [i \mu] - \ell^{acv}_{L^1}[\mu i] = 0, \quad \forall \mu \in \mc{N}, \forall (a,v) \in {L^1}, \forall c \in C , \\
	\sum_{i \in N \cup \mc{N}} \ell^{ab}_{L^2} [i \mu] - \ell^{ab}_{L^2}[\mu i] = 0, \quad \forall \mu \in \mc{N}, \forall (a,b) \in {L^2} ,
\end{gather}
which states that sum of all flows into and out of the physical intermediate node $\mu$ must be zero.

The actual amount of bandwidth reserved on a physical link $(\mu,\nu)$ after considering overlaps of $L^1$ flows can be captured by the following constraint:
\begin{equation}
	\sum_{\substack{a \in B, c \in C' \\ (c,v) \in E}} \ell^{acv}_{L^1} [\mu \nu] \le \ell^{v}_o [\mu \nu] , \quad \forall (\mu,\nu) \in \mc{E}, \forall (a,v) \in L^1, \forall C' \in C, |C'| \le k .
\end{equation}
The subset of critical nodes $C'$ represent a possible failure scenario where at most $k$ critical nodes fail. The RHS captures the maximum bandwidth used in those cases. Unfortunately, the caveat here is that this leads to an exponential expansion of constraints when $k$ goes large. The impact of overlapping redundant links, however, is significant as can be observed in the simulations in the next section.

The last set of constraints defines the link capacity on physical links $(\mu, \nu)$ and mapping links $(u, \mu)$.
\begin{gather}
	\begin{split}
		\sum_{(u,v) \in E} \bigl[ \ell^{uv}_E [\mu \nu] + \ell^{uv}_E [\nu \mu] \bigr] + \sum_{v \in N} \bigl[ \ell^{v}_o [\mu \nu] + \ell^{v}_o [\nu \mu] \bigr] \\
		\mbox{} + \sum_{a,b \in B} \bigl[ \ell^{ab}_{L^2} [\mu \nu] + \ell^{ab}_{L^2} [\nu \mu] \bigr] \le \Lambda_{\mu \nu} , \quad \forall (\mu, \nu) \in \mc{E} ,
	\end{split} \\
	\begin{split}
		\sum_{(u,v) \in E} \bigl[ \ell^{uv}_E [u \mu] + \ell^{uv}_E [\mu u] \bigr] + \sum_{v \in N} \bigl[ \ell^{v}_o [u \mu] + \ell^{v}_o [\mu u] \bigr] \\
		\mbox{} + \sum_{a,b \in B} \bigl[ \ell^{ab}_{L^2} [u \mu] + \ell^{ab}_{L^2} [\mu u] \bigr] \le \Lambda \rho_{u \mu} , \quad \forall (u, \mu) \in N \times \mc{N} .
	\end{split}
\end{gather}
The first constraint accounts for all flows on a physical link $(\mu,\nu)$ in {\em both} directions, and the total should be less than the physical remaining bandwidth $\Lambda_{\mu \nu}$. For the second constraint, the LHS is similar in that it accounts for all flows on the mapping link $(u, \mu)$, and $\Lambda$ is an arbitrary large constant. This way, the mapping variable $\rho_{u \mu}$ will be set to 1 if there are non-zero flows on that link in either direction.

\subsection{Objective Function and Approximation} 
\label{sec:object}

We seek to minimize the amount of resources used for a VInf. The objective function of the adapted MCF is then
\begin{multline}
	\min\quad \sum_{\mu \in \mc{N}} \alpha_\mu \sum_{u \in N \cup B} \rho_{u \mu} \gamma_\mu
	+ \sum_{(\mu,\nu) \in \mc{E}} \beta_{\mu\nu} \times 
\Biggl[ \sum_{v \in N} \bigl[ \ell^v_o [\mu \nu] + \ell^v_o [\nu \mu] \bigr] \\
		 \mbox{} + \sum_{a,b \in B} \bigl[ \ell_{L^2}^{ab}[\mu \nu] + \ell_{L^2}^{ab}[\mu \nu] \bigr] + \sum_{(u,v) \in E} \bigl[ \ell^{uv}_E [\mu \nu] + \ell^{uv}_E [\nu \mu] \bigr] \Biggr]
\label{eqn:objective}
\end{multline}
where $\alpha_\mu$ and $\beta_{\mu \nu}$ are node and link weights, respectively. To achieve load balancing across time, the weights can be set as $\frac{1}{\Gamma_\mu + \epsilon}$ and $\frac{1}{\Lambda_{\mu \nu} + \epsilon}$, respectively.

The variables of this linear program are the non-zero real-valued flows $\ell$ and the boolean mapping variables $\rho$. The presence of the boolean variables turns the linear program into a NP-Hard problem. An alternative is to relax the boolean variables to real-valued variables, obtain an approximate virtual node embedding by picking a map with the largest $\rho_{u \mu}$, and re-run the same linear program with the virtual nodes assigned to obtain the link assignments \cite{VineYard09}.




\section{Evaluation}
\label{sec:eval}

In this section, we evaluate the performance of the system when allocating resources with and without redundancy pooling and redundant bandwidth reduction, labeled {\bf share} and {\bf noshare} respectively. In particular, we focus on the resource utilization of the physical infrastructure and the admission rates of VInf requests. We further compare that to a system where VInfs do not have reliability requirement, i.e. zero redundancy (labeled {\bf nonr}), as a baseline to gauge the additional amount of resources consumed for reliability.

Our simulation setup is as follows. The physical infrastructure consists of 40 compute nodes with capacity uniformly distributed between 50 and 100 units. These nodes are randomly connected with a probability of 0.4 occurring between any two nodes, and the bandwidth on each physical link is uniformly distributed between 50 and 100 units. VInf requests arrive randomly over a timespan of 800 time slots and the inter-arrival time is assumed to follow the Geometric distribution at a rate of 0.75 per time slot. The resource lease times of each VInf follows the Geometric distribution as well at a rate of 0.01 per time slot. A high request rate and long lease times ensures that the physical infrastructure is operating at high utilization. Each VInf consists of nodes between 2 to 10, with a compute capacity demand of 5 to 20 per node. Up to 90\% of these nodes are critical and all failures are independent with probability 0.01. Connectivity between any two nodes in the VInf is random with probability 0.4, and the minimum bandwidth on any virtual link is 10 units. There are two main sets of results: (i) scaling the maximum bandwidth of a virtual link from 20 to 40 units while reliability guarantee of every VInf is $99.99\%$, and (ii) scaling the reliability guarantee of each VInf from 99.5\% to 99.995\% while the maximum bandwidth of a virtual link is 30 units. A custom discrete event simulator written in Python is used to run this setup on the Amazon EC2 platform\cite{AmazonEC2}, and the relaxed mixed integer programs are solved using the open-source CBC solver~\cite{coinCBC}.

\begin{figure*}
	\centering
	\newcommand\figsheight{1.71in}
	\newcommand\shorten{\vspace*{-8pt}}
	\subfloat[]{%
		\includegraphics[height=\figsheight]{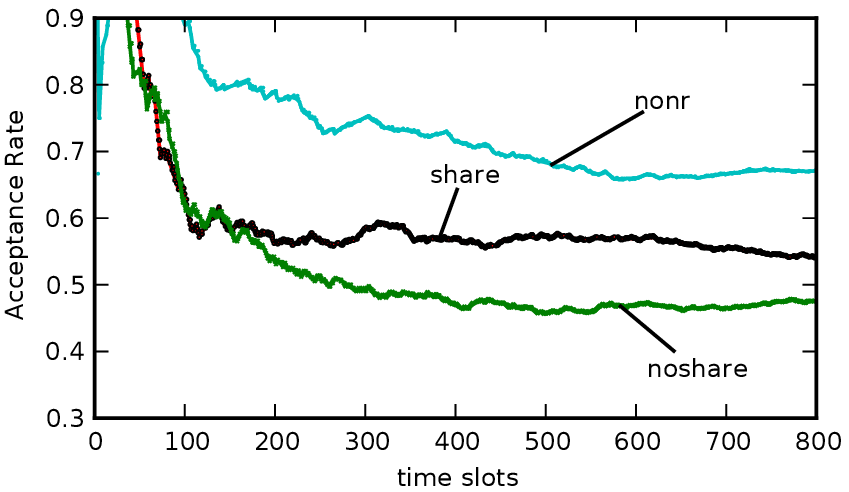}} \hfil
	\subfloat[]{%
		\includegraphics[height=\figsheight]{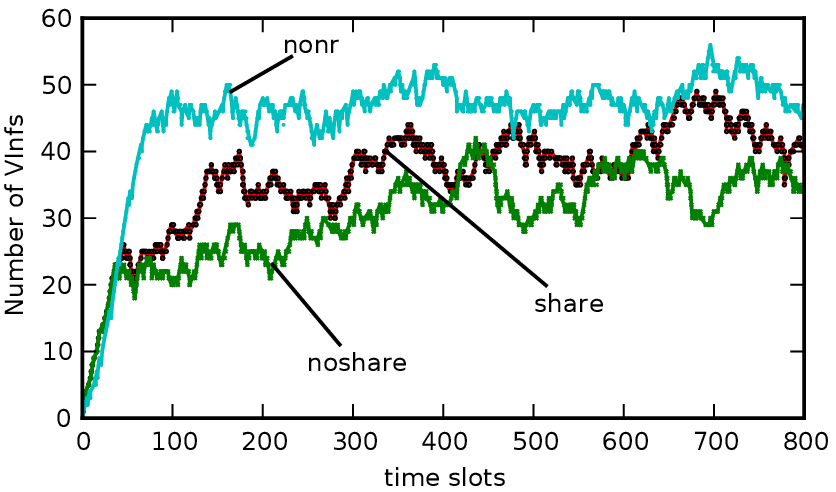}} \shorten\\
	\subfloat[]{%
		\includegraphics[height=\figsheight]{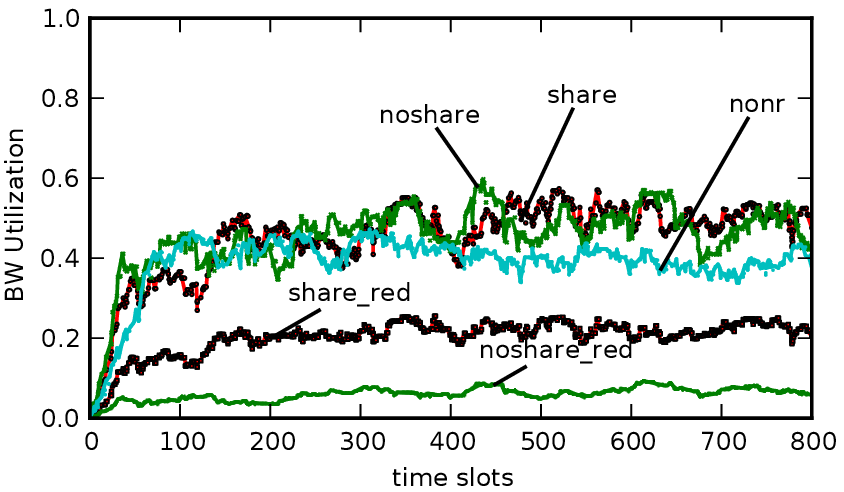}} \hfil
	\subfloat[]{%
		\includegraphics[height=\figsheight]{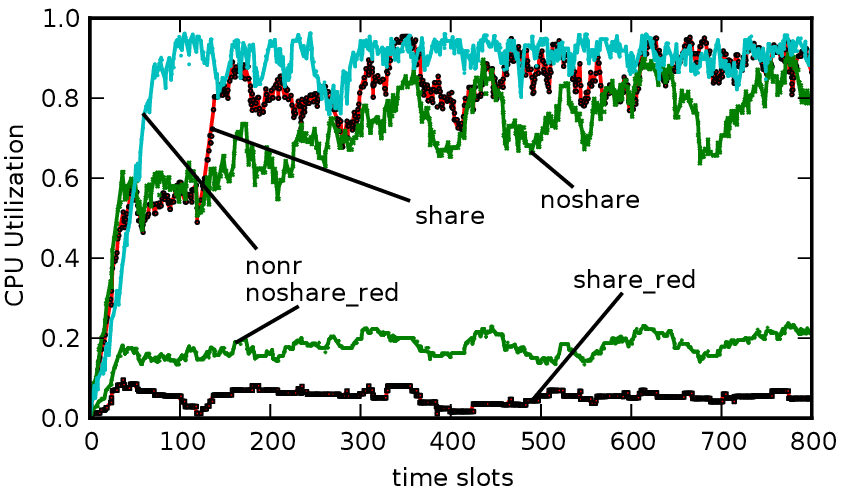}} \shorten\\
	\caption{State of the physical infrastructure over time in a single simulation instance where $r=99.99\%$ and max bandwidth per virtual link is 35.}
	\label{fig:compareShare}
\end{figure*}

\autoref{fig:compareShare} shows the acceptance rate, number of VInfs admitted, total bandwidth and CPU utilization (including redundancies) over time in a single simulation instance where reliability guarantee is 99.99\% and maximum bandwidth per virtual link is 35 units. The physical infrastructure is saturated by time later than 300s. As expected, the acceptance rate and the VInfs occupancy of the three cases {\bf nonr}, {\bf share} and {\bf noshare} are decreasing in that order, indicating that {\bf share} is more efficient in utilizing physical resources than {\bf noshare}. It can be seen that the increase in admitted VInfs is much slower in {\bf noshare} than the other two, and operates with lower resource utilization (especially for CPU). This indicates that {\bf noshare} is highly inefficient --- expansion of redundant backups and links without pooling have led to larger granularity in VInf resource requests. In terms of CPU utilization, redundant nodes in {\bf share} consume less resource (suffix {\sf \_red}) than that in {\bf noshare} despite admitting more VInfs, due to redundancy pooling. As for bandwidth utilization, {\bf share} do not use more bandwidth than {\bf noshare} even though more VInfs are admitted in the former case. We can observe later in \autoref{fig:compareRBW} that the bandwidth utilization is actually smaller for {\bf share}. This is so even though much higher bandwidth is dedicated for redundancy (suffix {\sf \_red}) in {\bf share} than {\bf noshare}. In the latter case, the redundant links use less bandwidth due to lower acceptance rates for VInfs with critical nodes. This effect can be seen in \autoref{fig:rejProf} for the same parameters over 10 instances.

\begin{figure}
	\centering
	\includegraphics[width=.6\linewidth]{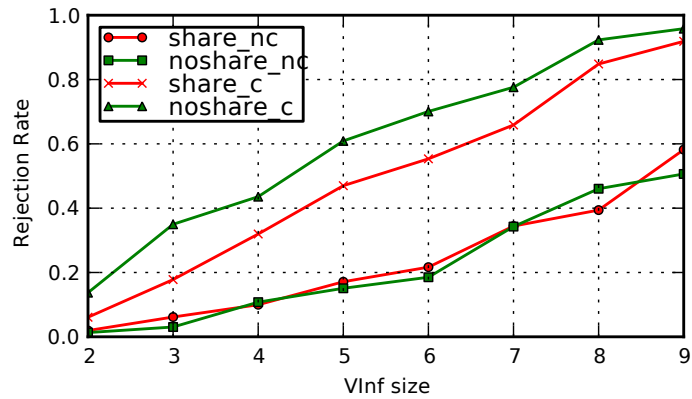}
	\caption{Rejection rates of VInfs according to the sizes. Suffix {\sf c} refers to VInfs with critical nodes ($r = 99.99\%$) and maximum bandwidth 35. {\sf nc} refers to those without.}
	\label{fig:rejProf}
\end{figure}

\autoref{fig:compareRBW} shows the mean performance of the three cases {\bf share}, {\bf noshare} and {\bf nonr} across 10 simulation runs. {\bf noshare} has the least acceptance rate and VInf occupancy, and more backup nodes than {\bf share}, which is able to pool redundancies and efficiently reuse backups. The impact of overlapping redundant links can be seen in \autoref{fig:compareRBW:bBk}. As bandwidth demands of VInf increases, {\bf noshare} rejects VInfs that have larger number of critical nodes, resulting in a sharper drop in backup nodes than {\bf share}. The latter conserves more bandwidth and is able to admit larger sized VInfs. In terms of resource utilization, similar trends as that of \autoref{fig:compareShare} can be observed across all reliability guarantees and maximum virtual link bandwidths.

In summary, increasing redundant nodes and expanding a VInf with backup links leads to VInfs with larger granularity. If the physical infrastructure admits these expanded VInfs verbatim as in the case of {\bf noshare}, much inefficiencies can occur. VInfs that have more nodes, bandwidth, or higher reliability requirement (or all of them) get expanded much larger than {\bf share}, leading to more rejections and losses in revenue. Smaller VInfs, especially those with no critical nodes, are more readily admitted and it is almost impossible for larger VInfs to be admitted. Although more CPU and bandwidth are utilized in the {\bf noshare} case, there is substantially less VInfs than {\bf share} (as much as 24\%) present in the physical infrastructure. This is so even where more bandwidth is dedicated for redundant links in {\bf share} as more VInfs with more critical nodes get admitted. In comparison to the case where no reliability is guaranteed ({\bf nonr}), the number of VInfs that can be admitted dropped by at most 20\% and the largest drop in acceptance rate goes from 65\% to 51\% when compared to {\bf share}. When compared to {\bf noshare}, the figures are 38\%, and from 65\% to 41\% respectively. Hence, the resources required for provisioning reliability is quite significant.

\newgeometry{top=.7in,bottom=1.2in,left=.72in,right=.72in}
\begin{figure*}
	\newcommand\figsheight{1.5in}
	\newcommand\shorten{\vspace*{-8pt}}
	\begin{minipage}{.48\linewidth}
		\raggedleft
		\subfloat[]{%
			\includegraphics[height=\figsheight]{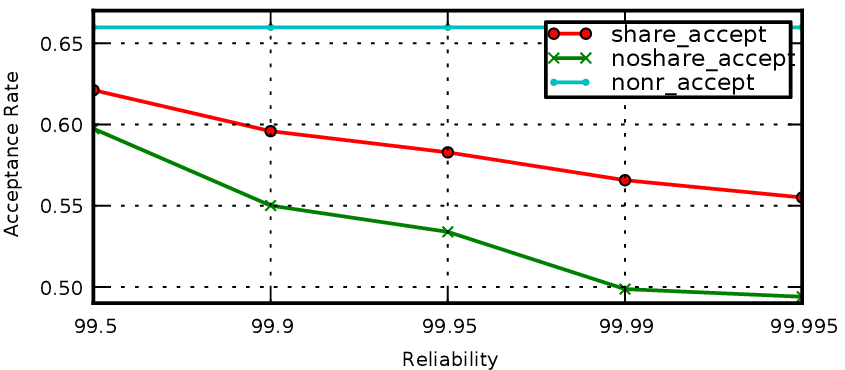}} \shorten \\
		\subfloat[]{%
			\includegraphics[height=\figsheight]{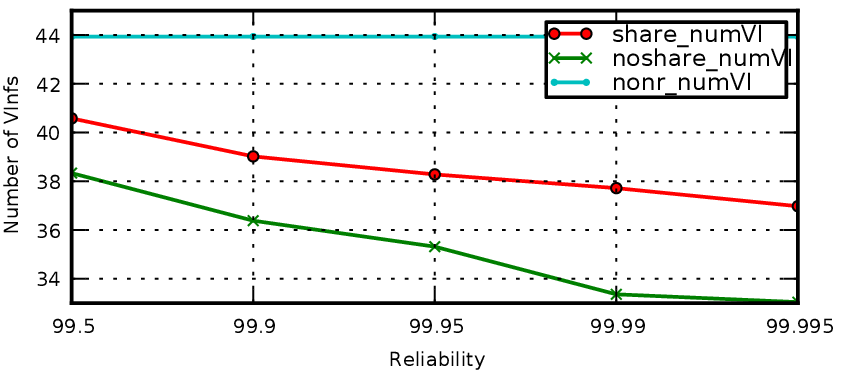}} \shorten \\
		\subfloat[]{%
			\includegraphics[height=\figsheight]{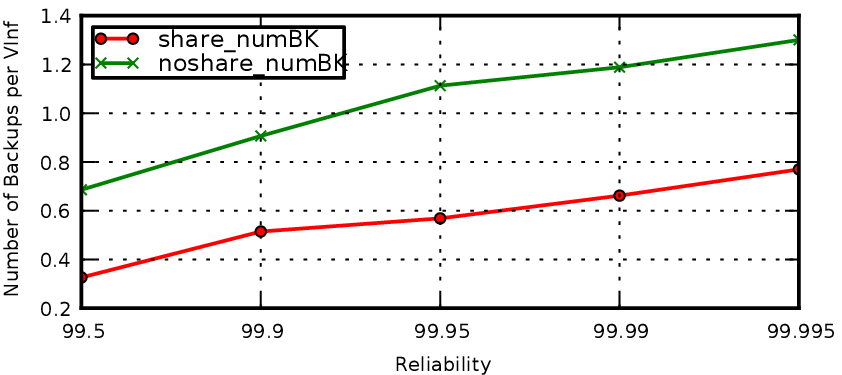}} \shorten \\
		\subfloat[]{%
			\includegraphics[height=\figsheight]{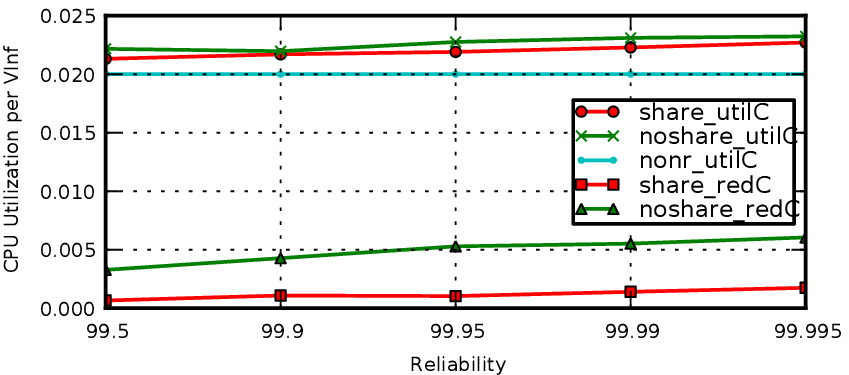}} \shorten \\
		\subfloat[]{%
			\includegraphics[height=\figsheight]{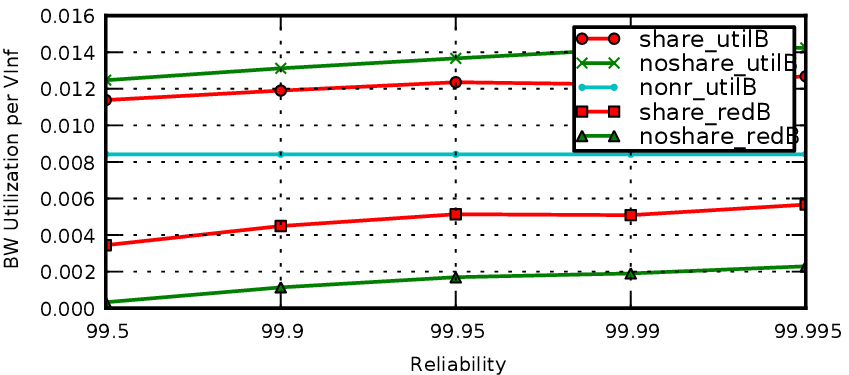}} \shorten \\
	\end{minipage} \hfil
	\begin{minipage}{.48\linewidth}
		\raggedleft
		\subfloat[]{%
			\includegraphics[height=\figsheight]{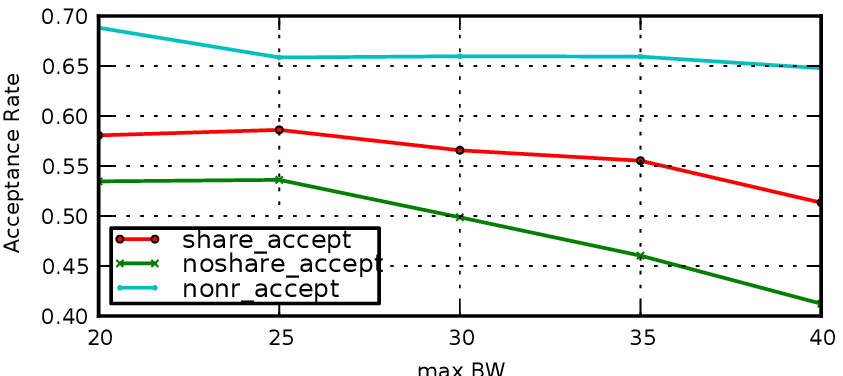}} \shorten \\
		\subfloat[]{%
			\includegraphics[height=\figsheight]{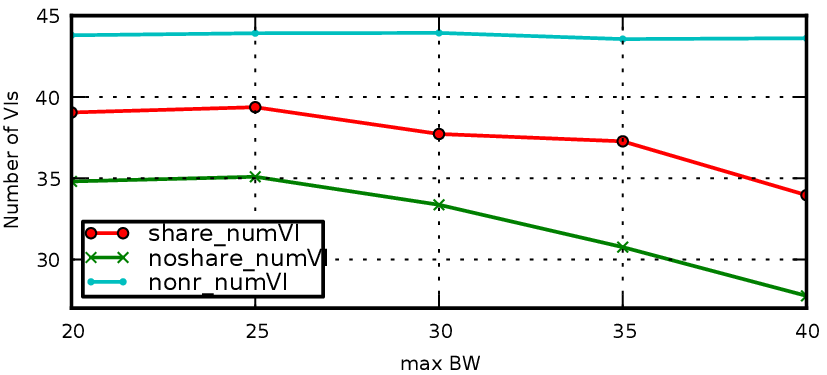}} \shorten \\
		\subfloat[\label{fig:compareRBW:bBk}]{%
			\includegraphics[height=\figsheight]{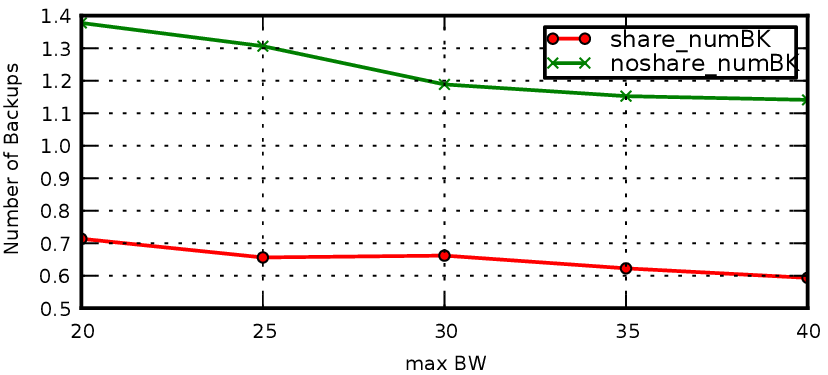}} \shorten \\
		\subfloat[]{%
			\includegraphics[height=\figsheight]{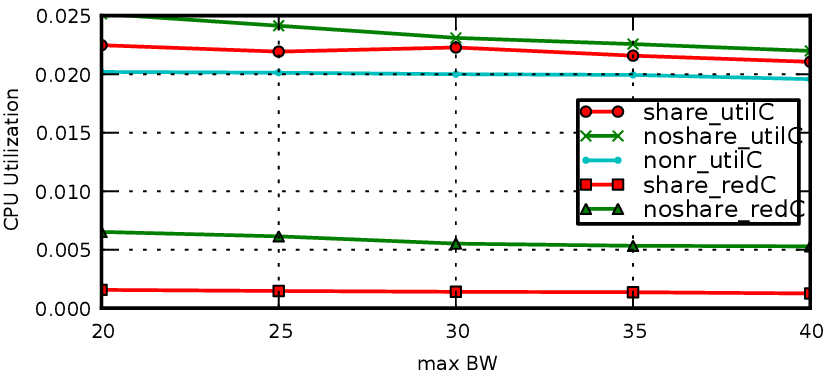}} \shorten \\
		\subfloat[]{%
			\includegraphics[height=\figsheight]{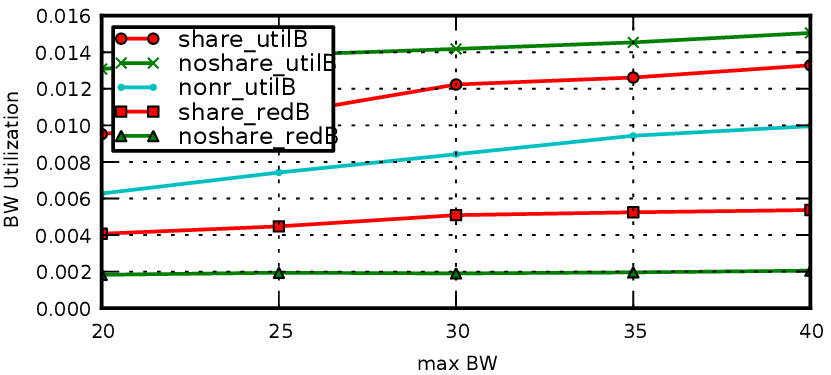}} \shorten \\
	\end{minipage}
	\caption{Figures above compare the performance of the system when backup nodes and links are shared against unshared. This is studied under varying reliability guarantees (left) and bandwidth requirements (right) of the VInf requests. {\bf nonr} is the baseline case where VInfs are admitted without any redundancy.}
	\label{fig:compareRBW}
\end{figure*}
\restoregeometry

\section{Related Work}
\label{sec:related}

Network virtualization is a promising technology to reduce the operating costs and management complexity of networks, and it is receiving an increasing amount of research interest~\cite{Chowdhury:stateart}. Reliability is bound to become a more and more prominent issue as the infrastructure providers move toward virtualizing their networks over simpler, cheaper commodity hardware~\cite{Trellis}.

Analysis on the reliability of overlay networks in terms of connectivity in the overlays has been developed~\cite{LeeKY10}. It achieves a good estimation of connectivity in the embedded overlay networks through a Monte Carlo simulation-based algorithm. Unfortunately, it is not applicable to our problem as we are concerned with critical virtual nodes and embedding them, as well as the whole infrastructure with reliability guarantees.

Fault tolerance is provided in data centers~\cite{Portland,Bcube} through special design of the network: having large excess of nodes and links in an organized manner as redundancies. These works provide reliability to the data center as a whole, but do not customize reliability guarantees to embedded virtual infrastructures.

While we are not aware of works studying the allocation of reliable virtual networks, \cite{debug} considered the use of ``shadow VNet'', namely a parallel virtualized slice, to study the reliability of a network. However, such slice is not used as a back-up, but as a monitoring tool, and as a way to debug the network in the case of failure. \cite{Wang08} considered the use of virtualized router as a management primitive that can be used to migrate routers for maximal reliability.

Meanwhile there are some works targeted at node fault tolerance at the server virtualization level. Bressoud~\cite{Bressoud96} was the first few to introduce fault tolerance at the hypervisor. Two virtual slices residing on the same physical node can be made to operate in sync through the hypervisor. However, this provides reliability against software failures at most, since the slices reside on the same node.

Others~\cite{Brendan08,Clark05} have made progress for the virtual slices to be duplicated and migrated over a network. Various duplication techniques and migration protocols were proposed for different types of applications (web servers, game servers, and benchmarking applications)~\cite{Clark05}. Remus~\cite{Brendan08} and Kemari~\cite{kemari08} are two other systems that allows for state synchronization between two virtual nodes for full, dedicated redundancy. However, these works focus on the practical issues, and do not address the resource allocation issue (in both compute capacity and bandwidth) while having redundant nodes residing somewhere in the network.

VNsnap~\cite{vnsnap09} is another method developed to take static snapshots of an entire virtual infrastructure to some reliable storage, in order to recover from failures. This can be stored reliably and distributedly as replicas~\cite{bigtable06}, or as erasure codes~\cite{DimakisAG06,Draid03}. There is no synchronization, and whether the physical infrastructure has sufficient resources to recover automatically using the saved snapshots is another question altogether.

At a fundamental level, there are methods to construct topologies for redundant nodes that address both nodes and links reliability~\cite{AjtaiM92,DuttS97}. Based on some input graph, additional links (or, bandwidth reservations) are introduced optimally such that the least number is needed. However, this is based on designing fault tolerance for multiprocessor systems which are mostly stateless. A node failure, in this case, involve migrations or rotations among the remaining nodes to preserve the original topology. This may not be suitable in a virtualized network scenario where migrations may cause much disruptions to parts of the network that are unaffected by the failure.

Our problem formulation involves virtual network embedding~\cite{VineYard09,IsoMorph09,VNetPathSplit08} with added node and link redundancy for reliability. In particular, our model employs the use of path-splitting~\cite{VNetPathSplit08}. Path-splitting is implicitly incorporated in our multi-commodity flow problem formulation. Path-splitting allows a flow between two nodes to be split over multiple routes such that the aggregate flow across those routes equal to the demand between the two nodes. This gives more resilience to link failures and allows for graceful degradation.

\section{Conclusion}
\label{sec:con}

We considered the problem of efficiently allocating resources in a virtualized physical infrastructure for Virtual Infrastructure (VInfs) with reliability guarantees, which is guaranteed through redundant nodes and links. Since a physical infrastructure hosts multiple VInfs, it is more resource efficient to share redundant nodes between VInfs. We introduced a pooling mechanism to share these redundancies for both independent and cascading types of failures. The physical footprint of redundant links can be reduced as well, by considering the maximum over all failure scenarios while allocating resources with a linear program adapted from the Multi-Commodity Flow problem. Both mechanisms have significant impact in conserving resources and improving VInf acceptance rates.


\bibliographystyle{abbrv}
\bibliography{IEEEabrv,paper}

\begin{thebibliography}{10}

\bibitem{handbook}
M.~Abramowitz and I.~A. Stegun.
\newblock {\em Handbook of Mathematical Functions with Formulas, Graphs, and
  Mathematical Tables}.
\newblock Dover, New York, ninth dover printing, tenth gpo printing edition,
  1964.

\bibitem{AjtaiM92}
M.~Ajtai, N.~Alon, J.~Bruck, R.~Cypher, C.~Ho, M.~Naor, and E.~Szemeredi.
\newblock Fault tolerant graphs, perfect hash functions and disjoint paths.
\newblock {\em Symposium on Foundations of Computer Science}, 0:693--702, 1992.

\bibitem{AmazonEC2}
{Amazon Elastic Compute Cloud (Amazon EC2)}.
\newblock \url{http://aws.amazon.com/ec2/}.

\bibitem{intel}
D.~Atwood and J.~G. Miner.
\newblock {Reducing Data Center Cost with an Air Economizer}.
\newblock
  \url{http://www.intel.com/it/pdf/Reducing\_Data\_Center\_Cost\_with\_an\_Air%
\_Economizer.pdf}, Aug. 2008.

\bibitem{Trellis}
S.~Bhatia, M.~Motiwala, W.~Muhlbauer, Y.~Mundada, V.~Valancius, A.~Bavier,
  N.~Feamster, L.~Peterson, and J.~Rexford.
\newblock Trellis: a platform for building flexible, fast virtual networks on
  commodity hardware.
\newblock In {\em CONEXT '08: Proceedings of the 2008 ACM CoNEXT Conference},
  pages 1--6, New York, NY, USA, 2008. ACM.

\bibitem{Bressoud96}
T.~C. Bressoud and F.~B. Schneider.
\newblock Hypervisor-based fault tolerance.
\newblock {\em ACM Trans. Comput. Syst.}, 14(1):80--107, 1996.

\bibitem{coinCBC}
{CBC: Coin-or Branch and Cut (COmputational INfrastructure for Operations
  Research)}.
\newblock \url{https://projects.coin-or.org/Cbc}.

\bibitem{bigtable06}
F.~Chang, J.~Dean, S.~Ghemawat, W.~C. Hsieh, D.~A. Wallach, M.~Burrows,
  T.~Chandra, A.~Fikes, and R.~E. Gruber.
\newblock {Bigtable: A Distributed Storage System for Structured Data}.
\newblock In {\em {Proc. OSDI 2006}}, Nov. 2006.

\bibitem{Chowdhury:stateart}
N.~M.~K. Chowdhury and R.~Boutaba.
\newblock Network virtualization: State of the art and research challenges.
\newblock {\em {IEEE Communication Magazine}}, 47(7):20--26, July 2009.

\bibitem{VineYard09}
N.~M. M.~K. Chowdhury, M.~R. Rahman, and R.~Boutaba.
\newblock {Virtual Network Embedding with Coordinated Node and Link Mapping}.
\newblock In {\em Proc. IEEE INFOCOM 2009}, Rio de Janeiro, Brazil, Apr. 2009.

\bibitem{Clark05}
C.~Clark, K.~Fraser, S.~Hand, J.~G. Hansen, E.~Jul, C.~Limpach, I.~Pratt, and
  A.~Warfield.
\newblock Live migration of virtual machines.
\newblock In {\em NSDI 2005: Proceedings of the 2nd USENIX Symposium on
  Networked Systems Design {\&} Implementation}. USENIX Association, May 2005.

\bibitem{Brendan08}
B.~Cully, G.~Lefebvre, D.~M.~M. Feeleyand, and N.~Hutchinson.
\newblock Remus: High availability via asynchronous virtual machine
  replication.
\newblock In {\em NSDI 2008: Proceedings of the 5th USENIX Symposium on
  Networked Systems Design {\&} Implementation}, pages 161--174. USENIX
  Association, Apr. 2008.

\bibitem{DimakisAG06}
A.~G. Dimakis, V.~Prabhakaran, and K.~Ramchandran.
\newblock {Decentralized erasure codes for distributed networked storage}.
\newblock {\em {IEEE} Trans. Inf. Theory}, 52(6):{2809--2816}, June 2006.

\bibitem{DobsonI05}
I.~Dobson, B.~A. Carreras, and D.~E. Newman.
\newblock A loading-dependent model of probabilistic cascading failure.
\newblock {\em Probability in the Engineering and Informational Sciences},
  19:15--32, 2005.

\bibitem{DuttS97}
S.~Dutt and N.~R. Mahapatra.
\newblock Node-covering, error-correcting codes and multiprocessors with very
  high average fault tolerance.
\newblock {\em IEEE Transactions on Computers}, 46(9):997--1015, 1997.

\bibitem{Bcube}
C.~Guo, G.~Lu, D.~Li, H.~Wu, X.~Zhang, Y.~Shi, C.~Tian, Y.~Zhang, and S.~Lu.
\newblock Bcube: A high performance, server-centric network architecture for
  modular data centers.
\newblock In {\em SIGCOMM '09: Proceedings of the ACM SIGCOMM 2009 conference
  on Data communication}, New York, NY, USA, 2009. ACM.

\bibitem{DiffEngine09}
D.~Gupta, S.~Lee, M.~Vrable, S.~Savage, A.~C. Snoeren, G.~Varghese, G.~M.
  Voelker, and A.~Vahdat.
\newblock {Difference Engine: Harnessing Memory Redundancy in Virtual
  Machines}.
\newblock {\em USENIX ;login}, 34(2), Apr. 2009.

\bibitem{HararyF96}
F.~Harary and J.~P. Hayes.
\newblock Node fault tolerance in graphs.
\newblock {\em Networks}, 27(1):19--23, 1996.

\bibitem{IyerSM09}
S.~M. Iyer, M.~K. Nakayama, and A.~V. Gerbessiotis.
\newblock A markovian dependability model with cascading failures.
\newblock {\em {IEEE} Trans. Comput.}, 58(9):1238--1249, 2009.

\bibitem{vnsnap09}
A.~Kangarlou, P.~Eugster, and D.~Xu.
\newblock {VNsnap: Taking Snapshots of Virtual Networked Environments with
  Minimal Downtime}.
\newblock In {\em Proc. IEEE/IFIP DSN 2009}, June 2009.

\bibitem{LeeKY10}
K.~Lee, H.-W. Lee, and E.~Modiano.
\newblock {Reliability in Layered Networks with Random Link Failures}.
\newblock In {\em Proc. IEEE INFOCOM 2010}, Mar. 2010.

\bibitem{IsoMorph09}
J.~Lischka and H.~Karl.
\newblock {A Virtual Network Mapping Algorithm based on Subgraph Isomorphism
  Detection}.
\newblock In {\em {VISA '09: Proceedings of the first ACM SIGCOMM workshop on
  Virtualized Infastructure Systems and Architectures}}, Barcelona, Spain, Aug.
  2009.

\bibitem{Draid03}
A.~D. Marco, G.~Chiola, and G.~Ciaccio.
\newblock {Using a Gigabit Ethernet Cluster as a Distributed Disk Array with
  Multiple Fault Tolerance}.
\newblock In {\em Proc. Local Networks (LCN 2003)}, Oct. 2003.

\bibitem{Portland}
R.~N. Mysore, A.~Pamboris, N.~Farrington, N.~Huang, P.~Miri, S.~Radhakrishnan,
  V.~Subramanya, and A.~Vahdat.
\newblock Portland: A scalable fault-tolerant layer 2 data center network
  fabric.
\newblock In {\em SIGCOMM '09: Proceedings of the ACM SIGCOMM 2009 conference
  on Data communication}, New York, NY, USA, 2009. ACM.

\bibitem{reliabilitybook}
P.~D.~T. O'Connor, D.~Newton, and R.~Bromley.
\newblock {\em Practical reliability engineering}.
\newblock John Wiley and Sons, fourth edition, 2002.

\bibitem{kemari08}
Y.~Tamura, K.~Sato, S.~Kihara, and S.~Moriai.
\newblock {Kemari: VM Synchronization for Fault Tolerance}.
\newblock In {\em {USENIX '08 Poster Session}}, June 2008.

\bibitem{Wang08}
Y.~Wang, E.~Keller, B.~Biskeborn, J.~van~der Merwe, and J.~Rexford.
\newblock Virtual routers on the move: live router migration as a
  network-management primitive.
\newblock In {\em SIGCOMM '08: Proceedings of the ACM SIGCOMM 2008 conference
  on Data communication}, pages 231--242, New York, NY, USA, 2008. ACM.

\bibitem{WattsDJ02}
D.~J. Watts.
\newblock A simple model of global cascades on random networks.
\newblock {\em Proc. of the National Academic Sciences}, 99(9):5766--5771, Apr.
  2002.

\bibitem{debug}
A.~Wundsam, A.~Mehmood, A.~Feldmann, and O.~Maennel.
\newblock Network troubleshooting with shadow vnets.
\newblock In {\em SIGCOMM Posters \& Demos}, August 2009.

\bibitem{VNetPathSplit08}
M.~Yu, Y.~Yi, J.~Rexford, and M.~Chiang.
\newblock Rethinking virtual network embedding: substrate support for path
  splitting and migration.
\newblock {\em SIGCOMM Comput. Commun. Rev.}, 38(2):17--29, 2008.

\end{thebibliography}

\newpage
\appendix

\section{Backups for cascading failure models} 
\label{app:kbackup}

The number of backups inherently depends on the failure model of the virtual infrastructure, as expressed in the reliability equation derived in \autoref{sec:numbackup}
\begin{equation}
	\tag{\ref{eqn:simpReliability}}
	r(k) = \sum_{y=0}^{k} \sum_{x=0}^{\min(n,y)} {k \choose y-x} p^{y-x} (1-p)^{k-(y-x)} f(x) .
\end{equation}
The core idea is to compute the distribution of the number of failed nodes based on the failure model $f(x)$ and iteratively search for a minimum $k$ that satisfies the reliability target $r$. The distribution $f(x)$ of three cascading models are computed as below.

\subsection{Load-based Model} 

The model \cite{DobsonI05} consists of $n$ identical nodes with an initial load uniformly distributed between $L^\text{min}$ and $L^\text{max}$, and a failure threshold of $L^\text{fail}$ on each node. In the first round $i=0$, every node is loaded with some disturbance $D$. The number of node failures $M_i$ are noted and every surviving node is incremented with a load equals $M_i P$. The cycle then continues until there are no node failures.

We use the following result from \cite{DobsonI05} in computing the number of backups required to guarantee reliability $r$. The distribution of the number of failed nodes is
\begin{equation}
	\label{eqn:loadBal}
	f_\text{LOAD} (x) = \begin{cases}
		\displaystyle{n \choose x} \phi(d) (d+xp)^{x-1} (\phi(1-d-xp))^{n-x}, & 0 \le x < n \\
		\displaystyle 1 - \sum_{x=0}^{n-1} f_\text{LOAD} (x), & x = n ,
	\end{cases}
\end{equation}
where the normalized load increase $p = \frac{P}{L^\text{max} - L^\text{min}}$, the normalized initial disturbance $d = \frac{D + L^\text{max} - L^\text{fail}}{L^\text{max} - L^\text{min}}$, and the saturation function
\begin{equation}
	\phi (z) = \begin{cases}
		0, & z < 0 \\
		z, & 0 \le z \le 1 \\
		1, & z > 1 .
	\end{cases}
\end{equation}
The above assumes $0^0 \triangleq 1$ and $\frac{0}{0} \triangleq 1$.


\subsection{Tree-based Model} 
\label{sec:ptree}

This model \cite{IyerSM09} is a continuous-time Markov Chain (CTMC) with the cascading effect described by rate $\phi_{ij}$, where a node of category $i$ causes a node of category $j$ to fail instantaneously. The state of the system is a $n+1$-dimensional vector $\chi = (e, \chi_1, \ldots, \chi_n)$ that counts the number of failed nodes in each category $\chi_i$, and the operating environment $e$ the system is in. Operating environments are used to describe various possible loading to the system, and the rate of transiting from one environment $e$ to the next is $\nu_e$. All states are recurrent with node failure and repair rates $\lambda_{i,e}$ and $\mu_{i,e}$, respectively.

Although the model considers some redundancy per category, we cannot directly import the result as the redundant nodes are not shared across the $n$ categories. To adopt this model for our system in \autoref{sec:numbackup}, we first assume that there are no redundancies, and then evaluate the failure distribution $f_\text{TREE} (x) = \sum_{F_x \subseteq \{1, \ldots, n\}} Pr(F_x)$, which can be used in \eqref{eqn:simpReliability} for $k$ shared redundant nodes. Specifically, the maximum number of each node category $i$ is restricted to 1, and the repair rate is set to a value that best describes MTTR.

Essentially, each state of the CTMC directly corresponds to a possible failure scenario $F_x$, and the failure distribution $f_\text{TREE} (x)$. In other words, the probability of $x$ failed nodes is the sum of steady-state probabilities $\pi(\chi)$ of all CTMC state-vectors $\chi$ whose numeric elements sum to $x$, i.e.,
\begin{equation}
	f_\text{TREE} (x) = \sum_{\chi : x = \sum_{i=1}^n \chi_i} \pi(\chi)
\end{equation}
Finding the said probabilities reduces to computing the infinitesimal generator matrix $Q$ of the CTMC, which is given in Algorithm~4 of the reference \cite{IyerSM09}. The algorithm generates all cascading failure trees, computes the rate of transition from one state to another, and fills in entries of $Q$ matrix. The steady-state probabilities $\pi$ is then obtained through solving the linear equations $\pi^T Q = 0$ and $\pi^T ( 1, \ldots, 1 ) = 1$ using Singular Value Decomposition.


\subsection{Degree-based Model} 
\label{app:degree}

In this cascading failure model \cite{WattsDJ02}, every node has a failure threshold $\phi$ that is picked from some random distribution $r(\phi)$, where $\int_0^1 \! r(\phi) \, d\phi = 1$. The network is initially perturbed with a small number of node failures. The cascading effect is defined as follows: a node will fail if the fraction of its neighbors that have failed is more than its threshold $\phi$. The failures then propagate until the failure condition cannot be met on any surviving node.

The failures will, with finite probability, cascade into a network-wide failure if the average degree $z = \sum_d d p_d$ for some node degree distribution $p_d$ satisfies the following
\begin{align}
	\label{eqn:cascadeDeg}
	z &< \sum_{d} d(d-1) \rho_d p_d , \\
\text{where} \quad \rho_d &= \begin{cases}
	1, & d=0, \\
	\int_0^{1/d} \! r(\phi) \, d\phi, & d>0 .
\end{cases}
\end{align}
This happens with probability
\begin{equation}
	P_d = 1 - (1 - G_0(1) + G_0(H_1(1)))^d,
\end{equation}
when the initial failed node is of degree $d$, where $G_0 (s) = \sum_d \rho_d p_d s^d$ and $H_1(s)$ is the solution to
\begin{align}
	H_1(s) &= 1 - G_1(1) + s G_1(H_1(s)), \\
\text{and} \quad G_1(s) &= \frac{1}{z} \frac{\partial}{\partial s} G_0(s) .
\end{align}

Hence, if the condition \eqref{eqn:cascadeDeg} is satisfied, the probability that all $n$ nodes will fail in a cascade is $f_\text{DEGREE}(n) = \sum_d P_d p_d$. Since the distribution of $f_\text{DEGREE} (x)$ is unknown for $x < n$, we use the worst case distribution $f'_\text{DEGREE} (x)$ into \eqref{eqn:simpReliability}, i.e,
\begin{equation}
	f'_\text{DEGREE} (x) = \begin{cases}
		1 - f_\text{DEGREE} (n), & x = n-1 \\
		f_\text{DEGREE} (n), & x = n \\
		0, & \text{otherwise.}
	\end{cases}
\end{equation}
Unfortunately, the minimum $k$ evaluated through the above will never be less than $n-1$, which gives higher reliability than the required $r$. The margin closes when the probability of a cascading failure is high. To get a tighter value of $k$, a Monte-Calo simulation could be performed to estimate the distribution $f(x)$.


\subsection{Numerical method for $k$} 
\label{app:iterate}

Evaluating for $k$ equates to solving the inverse of \eqref{eqn:simpReliability}. We briefly describe a numerical method, which makes use of two properties of $r$, to accomplish the task.
\begin{enumerate}
	\item $r$ increases monotonically as the number of redundant nodes $k$ increases for a given failure distribution $f(x)$, and
	\item the maximum value of $k_\text{max}$ for any failure distribution $f(x)$ is the smallest integer that satisfy
		\begin{equation}
			\label{eqn:kmax}
			r(k_\text{max}) \le \sum_{y=0}^{k_\text{max}-n} { k_\text{max} \choose y } p^y (1-p)^{k_\text{max}-y}
		\end{equation}
\end{enumerate}
The latter is the worst case scenario where all $n$ virtual nodes fail with probability 1. Then, $k$ should be large enough such that, with probability $r$, there are $n$ surviving nodes. A straightforward numerical method is to use a classical binary search for $k$ as shown in \autoref{alg:binsearchk}. Another option would be to do a gradient descent of \eqref{eqn:simpReliability} on $k$.

\begin{algorithm}
	\caption{Search for $k$}
	\label{alg:binsearchk}
	\begin{algorithmic}[1]
		\State $k \gets$ BinSearchK($r$, 0, $k_\text{max}$)
		\State
		\Procedure{BinSearchK}{$r$, $\underline{k}$, $\overline{k}$}
		\If {$\underline{k} + 1 \ge \overline{k}$}
			\If {$r(\underline{k}) < r$}
				\State \Return $\overline{k}$
			\Else
				\State \Return $\underline{k}$
			\EndIf
		\Else
			\State $k' \gets \lceil \frac{1}{2} (\underline{k} + \overline{k}) \rceil$
			\State $r' \gets r(k')$ 
			\If { $r' < r$ }
				\State \Return BinSearchK($r$, $k'+1$, $\overline{k}$)
			\Else
				\State \Return BinSearchK($r$, $\underline{k}$, $k'$)
			\EndIf
		\EndIf
		\EndProcedure
	\end{algorithmic}
\end{algorithm}

The time complexity to compute $k$ is $O(n^2 \log n)$ with the procedure listed in \autoref{alg:binsearchk}.
\begin{proof}
	From \autoref{thm:linearK} and \eqref{eqn:kmax}, $k_\text{max} = O(n)$. Hence, the number of calls to the BinSearchK procedure in \autoref{alg:binsearchk} is $O(\log n)$. The time complexity in evaluating the reliability of every intermediate point $k'$ using \eqref{eqn:simpReliability} in the worst case is $O(k^2)$.
	Invoking \autoref{thm:linearK} again gives us the combined time complexity to be $O(n^2 \log n)$.
\end{proof}



\end{document}